\documentclass[12pt]{article}

\usepackage{geometry}
\usepackage{setspace}

\usepackage{lmodern,enumerate,rotating,anyfontsize}
\usepackage{amsmath,amsfonts,amssymb,subfigure,mathtools,blkarray,soul,etoolbox}
\usepackage{tikz-cd}
\usetikzlibrary{cd}
\usetikzlibrary{shapes,shapes.geometric,arrows,fit,calc,positioning,automata}

\usepackage{stmaryrd}
\SetSymbolFont{stmry}{bold}{U}{stmry}{m}{n}

\DeclareMathOperator{\CCa}{CC_A}

\DeclareMathOperator{\obs}{obs}

\DeclareMathOperator{\UR}{UR}
\DeclareMathOperator{\dss}{dss}
\DeclareMathOperator{\TwPl}{TwPl}
\DeclareMathOperator{\Ver}{Ver}

\newcommand{\N}{\mathbb{N}}

\newcommand{\LM}{\mathcal{L}}

\newcommand{\dt}{\delta}

\newcommand{\ep}{\epsilon}
\newcommand{\vep}{\varepsilon}

\newcommand{\Scal}{\mathcal{S}}

\newcommand{\Sig}{\Sigma}
\newcommand{\s}{\sigma}

\newcommand{\Mt}{\mathcal{M}}

\newcommand{\nsf}{\mathsf{n}}
\newcommand{\fsf}{\mathsf{f}}
\newcommand{\Ef}{{\cyan E_f}}
\newcommand{\ef}{{\cyan e_f}}
\newcommand{\QS}{{\red Q_S}}
\newcommand{\f}{{\cyan f}}
\newcommand{\fs}{{\cyan\{f\}}}

\newtheorem{theorem}{Theorem}[section]

\newtheorem{corollary}[theorem]{Corollary}
\newtheorem{definition}{Definition}
\newtheorem{proposition}[theorem]{Proposition}

\newtheorem{example}{Example}

\newtheorem{remark}{Remark}

\usepackage{times,amssymb,amsfonts,mathtools,graphicx,tikz,algorithm,algorithmic,url,hhline,booktabs}
\usetikzlibrary{automata,arrows,positioning}

\usetikzlibrary{shadows}

\makeatletter
\let\NAT@parse\undefined
\makeatother
\usepackage[colorlinks,linkcolor=blue,anchorcolor=blue,citecolor=green,urlcolor=cyan,hyperindex]{hyperref}

\newcommand{\PSPACE}{\mathsf{PSPACE}}

\usetikzlibrary{automata,arrows,positioning}
\usetikzlibrary{petri}

\newcounter{enumi_saved}

\tikzset{elliptic state/.style={draw,ellipse, thick, fill=gray!10}
}

\tikzset{
node distance=3cm, 
every state/.style={thick, fill=gray!10}, 
initial text=$ $, 
}

\newcommand{\red}{\color{red}}

\definecolor{green}{rgb}{0.1,0.7,0.1}

\newcommand{\cyan}{\color{cyan}}

\usepackage{tcolorbox}
\usepackage{tabularx}
\usepackage{array}
\usepackage{colortbl}
\tcbuselibrary{skins}

\title{A unified concurrent-composition method to state/event inference and concealment 
in discrete-event systems
}

\author{Kuize Zhang\\
{\small Control Systems Group, Technical University of Berlin, 10587 Berlin, Germany}\\
{\small kuize.zhang@campus.tu-berlin.de}
	}

\begin{document}

\date{}

\maketitle

{\bf Abstract}
Discrete-event systems usually consist of discrete states and transitions between them caused by 
spontaneous occurrences of labelled (aka partially-observed) events. Due to the partially-observed feature,
fundamental properties therein could be classified into two categories: state/event-inference-based properties 
(e.g., strong detectability, diagnosability, and predictability) and state-concealment-based properties (e.g.,
opacity). Intuitively, the former category describes whether one can use observed output sequences to infer 
the current and subsequent states, past occurrences of faulty events, or future certain occurrences of faulty 
events; while the latter describes whether one cannot use observed output sequences to infer whether some 
secret states have been visited (that is, whether the DES can conceal the status that its secret states have been
visited). Over the past two decades these properties were studied separately using different methods.
In this review article, for labeled finite-state automata, a unified 
concurrent-composition method is shown to verify all above inference-based properties and concealment-based
properties, resulting in a unified mathematical framework for the two categories of properties. In addition,
compared with the previous methods in the literature, the concurrent-composition method does not depend on
assumptions and is more efficient.


\section{Introduction}

\emph{Discrete-event systems} (DESs) usually consist of discrete states and transitions between them caused by
spontaneous occurrences of labelled (aka partially-observed) events. Hence DESs are autonomous (they are not
driven by external factors) 
and nonlinear \cite{WonhamSupervisoryControl}. Due to the partially-observed feature of DESs, observation-based
fundamental properties could be classified into two categories: \emph{inference}-based and \emph{concealment}-based.
The former
means whether one can infer further information of a DES from the observations to the DES, so that
these information could be used to do further study on the DES, e.g., synthesizing a controller to change several
properties of the DES. While the latter refers to whether the DES can forbid several information from being leaked to
external intruders, even though the intruders can see outputs/labels generated by the DES. Roughly speaking, the former 
is dual to the latter. The former category contain detectability, diagnosability, predictability, etc.;
the latter contain opacity, etc. Detectability, diagnosability, and predictability mean whether one can use
observed output sequences generated by a DES to determine its current states, past occurrences of faulty events,
and future certain occurrences of faulty events; opacity refers to whether one cannot use 
observed output sequences to determine whether secret states have been visited.

Over the two past decades, these properties
were studied separately, but no intrinsic relations between them were revealed. In this paper, for DESs
modeled by labeled finite-state automata (LFSAs), a unified 
\emph{concurrent-composition} method is given to verify all of them, thus a unified mathematical framework is 
given to include many inference-based and concealment-based properties. Roughly speaking, a concurrent composition 
collects all pairs of trajectories of two systems producing the same output sequence, in which observable 
transitions with the same output are synchronized but unobservable transitions interleave. For strong detectability,
diagnosability, and predictability (which are inference-based
properties), the concurrent compositions of trivial variants of an LFSA and trivial variants of itself are
constructed to do verification in polynomial time; for variants of opacity (which are concealment-based properties), 
the concurrent compositions of trivial variants of an LFSA and the observers\footnote{i.e., the standard
powerset construction used for determinizing nondeterministic finite automata with $\ep$-transitions
\cite{Sipser2006TheoryofComputation}} of trivial variants of itself
are constructed to do verification in exponential time. In a large extent, \emph{from a complexity point of view, it is easy to
infer something, but it is hard to conceal something.}


The concurrent-composition method shows advantage in verifying inference-based properties of DESs since it does 
not rely on any assumption. Existing results in employing the detector method \cite{Shu2011GDetectabilityDES} to
verify strong detectability depends on two fundamental assumptions, i.e., deadlock-freeness which assumes an LFSA
always running, and divergence-freeness which requires the running of an LFSA will always be eventually observed.
These two assumptions are also adopted in the twin-plant method \cite{Jiang2001PolyAlgorithmDiagnosabilityDES} and
the verifier method in \cite{Yoo2002DiagnosabiliyDESPTime,Genc2009PredictabilityDES} for verifying diagnosability
and predictability.

When being applied to verify concealment-based properties, the concurrent-composition method is more efficient
than almost all methods in the literature, e.g., the initial-state-estimator method 
\cite{Saboori2012InfiniteStepOpacity}, the two-way-observer method \cite{Yin2017TWObserverInfiniteStepOpacity},
the $K$-delay-trajectory-estimator-method \cite{Falcone2015StrongKStepOpacity}, and the $K$/Inf-step-recognizer 
method \cite{Ma2021StrongInfSODES}.
The only known exception lies in the fact that only two variants of opacity, current-state opacity and
initial-state opacity, could be verified by directly using
the notions of observer \cite{Cassez2009DynamicOpcity,Saboori2007CurrentStateOpacity} and reverse observer
\cite{Wu2013ComparativeAnalysisOpacity} (see Section~\ref{sec:StanOpa} for details).

{\bf Notation} 
Symbol $\N$ denotes the set of nonnegative integers.
For an (finite) alphabet $\Sig$ (i.e., every sequence of elements of $\Sig$ is a unique sequence of elements of
$\Sig$, e.g., $\{0,00\}$ is not an alphabet since $000=0\ 00=00\ 0$), 
$\Sig^*$ and $\Sig^{\omega}$ are used to denote the set of
\emph{words} (i.e., finite-length sequences of elements of $\Sig$) over $\Sig$ including the empty word $\epsilon$
and the set of \emph{configurations} (i.e., infinite-length sequences of elements of $\Sig$) over $\Sig$,
respectively. $\Sig^{+}:=\Sig^*\setminus\{\epsilon\}$.
For a word $s\in \Sig^*$,
$|s|$ stands for its length, and we set $|s'|=+\infty$ for all $s'\in \Sig^{\omega}$.
For $s\in \Sig^+$ and $k\in\N$, $s^k$ and $s^{\omega}$ denote the concatenations of $k$ copies of
$s$ and infinitely many copies of $s$, respectively. Analogously, $L_1L_2:=\{e_1e_2|e_1\in L_1,
e_2\in L_2\}$, where $L_1,L_2\subset \Sig^*$.
For a word (configuration) $s\in \Sig^*(\Sig^{\omega})$, a word $s'\in \Sig^*$ is called a \emph{prefix} of $s$,
denoted as $s'\sqsubset s$,
if there exists another word (configuration) $s''\in \Sig^*(\Sig^{\omega})$ such that $s=s's''$.
For a set $S$, $|S|$ denotes its cardinality and $2^S$ its power set. Symbols $\subset$
and $\subsetneq$ denote the subset and strict subset relations, respectively.

\begin{definition}\label{FA:def8} 
		A \emph{labeled finite-state automaton}\index{labeled finite-state automaton} is a sextuple
		$\Scal=(Q,E,\dt,Q_0,\Sig,\ell)$, where
		\begin{enumerate}
			\item $Q$ is a finite set of \emph{states},
			\item $E$ (which is an alphabet) is a finite set of \emph{events},
			\item $\dt:Q\times E \to 2^Q$ is the \emph{transition function} (equivalently represented by
				the \emph{transition relation} $\dt\subset Q\times E\times Q$ such that $q'\in\dt(q,e)$ if and 
				only if $(q,e,q')\in\dt$),
			\item $Q_0\subset Q$ is a set of \emph{initial states}, 
			\item $\Sig$ (also an alphabet) is a finite set of \emph{outputs}/\emph{labels}, and
			\item $\ell: E\to \Sig\cup\{\ep\}$ is the \emph{labeling function}.
		\end{enumerate}
	\end{definition}

A transition $(q,e,q')\in\dt$ is interpreted as when $\Scal$ is in state $q$ and event $e$ occurs $\Scal$
transitions to state $q'$. 
The event set $E$ can been rewritten as disjoint union
of \emph{observable} event set $E_o=\{e\in E|\ell(e)\in\Sig\}$ and \emph{unobservable} event set
$E_{uo}=\{e\in E|\ell(e)=\ep\}$. When an observable event occurs, its label can be observed; when an
unobservable event occurs, nothing can be observed. Transition function $\dt:Q\times E\to 2^Q$ is recursively
extended to $\dt:Q\times E^*\to 2^Q$ as follows: for all $q\in Q$, $u\in E^*$, and $a\in \Sig$, one has
$\dt(q,\ep)=\{q\}$ and $\dt(q,ua)=\bigcup_{p\in \dt(q,u)}\dt(p,a)$; equivalently, transition relation
$\dt\subset Q\times E 
\times Q$ is recursively extended to $\dt\subset Q\times E^* \times Q$ as follows:
(1) for all $q,q'\in Q$, $(q,\ep,q')\in\dt$ if and only if $q=q'$; (2) for all $q,q'\in Q$,
$s\in E^*$, and $e\in E$, one has $(q,se,q')\in\dt$, also denoted by $q\xrightarrow[]{se}
q'$, called \emph{transition sequence} or \emph{run}\index{run}, if and only if, $(q,s,q''),(q'',e,q')\in\dt$
for some $q''\in Q$. 

Labeling function $\ell:E\to\Sig\cup\{\ep\}$ 
is recursively extended to $\ell:E^*\cup E^{\omega}\to\Sig^*\cup\Sig^{\omega}$ as
$\ell(e_1e_2\dots)=\ell(e_1)\ell(e_2)\dots$ and $\ell(\ep)=\ep$. 
For all $E'\subset E$, $\ell(E'):=\{\ell(e)|e\in E'\}$.
Transitions $x\xrightarrow[]{e}x'$ with $\ell(e)=\ep$ (resp., $\ell(e)\ne\ep$) are called 
\emph{unobservable} (resp., \emph{observable}).
For $q\in Q$ and $s\in E^+$, $(q,s,q)$ is called a \emph{transition 
cycle} if $(q,s,q)\in\dt$. An \emph{observable transition cycle} is defined 
by a transition cycle with at least one observable transition. Analogously an \emph{unobservable
transition cycle} is defined by a transition cycle with no observable transition.
An LFSA is called \emph{deterministic} if $|Q_0|=1$ and
for all $q,q',q''\in Q$ and $e\in E$, if $(q,e,q'),(q,e,q'')\in\dt$ then $q'=q''$.

A state $q\in Q$ is called \emph{live} if $(q,e,q')\in \dt$ for some $e\in E$ and $q'\in Q$.
$\Scal$ is called \emph{live}/\emph{deadlock-free} if each of its reachable states is live.
\emph{A state $q'\in Q$ is reachable from a state}
$q\in Q$ if there exists $s\in E^+$ such that $q\xrightarrow[]{s}q'$.
\emph{A subset $Q'$ of $Q$ is reachable from a state} $q\in Q$ if some state of $Q'$ is 
reachable from $q$. Similarly \emph{a state $q\in Q$ is reachable from a subset $Q'$ of $Q$}
if $q$ is reachable from some state of $Q'$. A state $q\in Q$ is called
\emph{reachable} (\emph{in $\Scal$}) if either $q\in Q_0$ or it is reachable from some initial state.
For a transition $(q,e,q')\in\dt$, a transition $(q'',e',q''')\in\dt$ is called a \emph{predecessor} of
$(q,e,q')$ if either $q=q'''$ or $q$ is reachable from $q'''$; a transition $(q'',e',q''')$ is called 
a \emph{successor} of $(q,e,q')$ if either $q''=q'$ or $q''$ is reachable from $q'$.

The symbol $L(\Scal):=\{s\in E^*|(\exists q_0\in Q_0)(\exists q\in Q)[q_0\xrightarrow[]
{s}q]\}$ will be used to denote the set of finite-length event sequences generated by $\Scal$,
$L^{\omega}(\Scal)=\{e_1e_2\dots $$ \in E^{\omega}|(\exists q_0\in Q_0)
(\exists q_1,q_2,\dots $$\in Q)[q_0\xrightarrow[]{e_1}q_1\xrightarrow[]{e_2}\cdots]\}$ will 
denote the set of infinite-length event sequences generated by $\Scal$.
For each $\s\in\Sig^*$, $\Mt(\Scal,\s)$ denotes the \emph{current-state estimate},
i.e., the set of states that the system can be in when $\s$ has just been generated, i.e., 
$\Mt(\Scal,\s):=\{q\in Q|(\exists q_0\in Q_0)(\exists s\in E^*)[(\ell(s)=\s)\wedge(q_0\xrightarrow[]{s}q)]\}$.
$\LM({\Scal})$ denotes the \emph{language generated} by $\Scal$,
i.e., $\LM({\Scal}):=\{\s\in\Sig^*|\Mt(\Scal,\s)\ne\emptyset\}$.
$\LM^{\omega}({\Scal})$ denotes the $\omega$-\emph{language generated} by $\Scal$,
i.e., $\LM^{\omega}(\Scal):=\{\s\in\Sigma^{\omega}|(\exists s\in L^{\omega}(\Scal))
[\ell(s)$$=\s]\}$. For a subset $x\subset Q$ of states, its \emph{unobservable reach} $\UR(x)$ is defined by
$\bigcup_{q\in x}\bigcup_{s\in (E_{uo})^*}\dt(q,s)$.

\begin{example}\label{FA:exam18} 
	Consider the following LFSA $\Scal_1$. It is deterministic but not live ($q_1$ is not live).
	One sees $L(\Scal_1)=\{(e_1)^n,(e_1)^ne_2|n\ge 0\}$, $L^{\omega}(\Scal_1)=\{(e_1)^{\omega}\}$,
	$\LM(\Scal_1)=\{a^n,a^nb|n\ge 0\}$, $\LM^{\omega}(\Scal_1)=\{a^{\omega}\}$, and $\Mt(\Scal_1,b)=\{q_1\}$.
	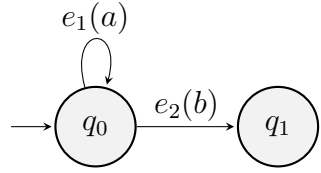
\begin{figure}[H]
		\centering
		\begin{tikzpicture}[>=stealth',shorten >=1pt,auto,node distance=2.4 cm, scale = 1.0, transform shape,
	>=stealth,inner sep=2pt,
		empty/.style={}]

	\node[initial, state] (s0) {$q_0$};
	\node[state] (s1) [right of =s0] {$q_1$};
	
	\path [->]
	(s0) edge [loop above] node [above, sloped] {$e_1(a)$} (s0);
	\path [->]
	(s0) edge node [above, sloped] {$e_2(b)$} (s1);
	;

		\end{tikzpicture}
		\caption{LFSA $\Scal_1$, where $q_0$ is the initial state (with an input arrow from nowhere),
		$\ell(e_1)=a$, $\ell(e_2)=b$.}
	\end{figure} 
\end{example}

Next we introduce the main tool --- concurrent composition. The concurrent-composition
structure exactly arose from characterizing \emph{negation} of a strong version of detectability called
\emph{eventual strong detectability} in
\cite{Zhang2020DetPNFA}, where the eventual strong detectability is essentially different from and strictly
weaker than the notion of \emph{strong
detectability} proposed in \cite{Shu2007Detectability_DES}. In the concurrent composition of two 
automata, observable transitions with the same label are synchronized, while unobservable transitions 
interleave.

\begin{definition}[\cite{Zhang2019KDelayStrDetDES,Zhang2020DetPNFA}]\label{FA:def9}
	Consider two LFSAs $\Scal^i=(Q_i,E,\dt_i,Q_{0i},\Sig,\ell)$, $i=1,2$, we define the \emph{concurrent
	composition} $\CCa(\Scal^1,\Scal^2)$ of $\Scal^1$ and $\Scal^2$ 
	by
		\begin{equation}\label{FA:eqn8}
			\CCa(\Scal^1,\Scal^2)=(Q',E',\dt',Q_0',\Sig,\ell'),
		\end{equation} where
\begin{enumerate}
	\item $Q'=Q_1\times Q_2$;
	\item $E'=E_o'\cup E_{uo}'$, where $E_o'=\{(\breve{e},\breve{e}')|\breve{e},\breve{e}'\in E_o,
		\ell(\breve{e})=\ell(\breve{e}')\}$,
		$E_{uo}'=\{(\breve{e},\epsilon)|\breve{e}\in E_{uo}\}\cup
		\{(\epsilon,\breve{e})|\breve{e}\in E_{uo}\}$;
	\item for all $(\breve{q}_1,\breve{q}_1'),(\breve{q}_2,\breve{q}_2')\in Q'$, $(\breve{e},\breve{e}')
		\in E_o'$, $(\breve{e}'',\epsilon)\in E_{uo}'$, and $(\epsilon,\breve{e}''')\in E_{uo}'$,
		\begin{itemize}
			\item $((\breve{q}_1,\breve{q}_1'),(\breve{e},\breve{e}'),(\breve{q}_2,\breve{q}_2'))\in\dt'$ 
				if and only if $(\breve{q}_1,\breve{e},\breve{q}_2)\in\dt_1$, $(\breve{q}_1',\breve{e}',\breve{q}_2')\in\dt_2$,
			\item $((\breve{q}_1,\breve{q}_1'),(\breve{e}'',\epsilon),(\breve{q}_2,\breve{q}_2'))\in\dt'$ 
				if and only if $(\breve{q}_1,\breve{e}'',\breve{q}_2)\in\dt_1$, $\breve{q}_1'=\breve{q}_2'$,
			\item $((\breve{q}_1,\breve{q}_1'),(\epsilon,\breve{e}'''),(\breve{q}_2,\breve{q}_2'))\in\dt'$ 
				if and only if $\breve{q}_1=\breve{q}_2$, $(\breve{q}_1',\breve{e}''',\breve{q}_2')\in\dt_2$;
		\end{itemize}
	\item $Q_0'=Q_{01}\times Q_{02}$;
	\item for all $(\breve{e},\breve{e}')\in E'$, $\ell'((\breve{e},\breve{e}')):=\ell(\breve{e})=\ell(\breve{e}')$.
	\end{enumerate}
	Particularly if $\Scal^1=\Scal^2$, then $\CCa(\Scal^1,\Scal^2)=:\CCa(\Scal^1)$ is called the 
	\emph{self-composition} of $\Scal^1$.
\end{definition}

For an event sequence $s'\in (E')^{*}$, we use $s'(L)$ and $s'(R)$ to
denote its left and right
components, respectively. Similar notation is applied to states of $Q'$. 
In addition, for every $s'\in(E')^{*}$, we use $\ell(s')$
to denote $\ell(s'(L))$ or $\ell(s'(R))$, since $\ell(s'(L))=\ell(s'(R))$. 
In the above construction, $\CCa(\Scal^1,\Scal^2)$ aggregates all pairs of runs of 
$\Scal_1$ and runs of $\Scal_2$ that produce the same label sequence. 

\begin{example}\label{FA:exam19} 
	An LFSA $\Scal_2$ and its self-composition $\CCa(\Scal_2)$ are
	shown in Figure~\ref{FA:fig26}.
	\begin{figure}[H]
     \centering
	 \subfigure[$\Scal_2$.]{
\begin{tikzpicture}[>=stealth',shorten >=1pt,auto,node distance=2.5 cm, scale = 1.0, transform shape,
	>=stealth,inner sep=2pt]

	\node[initial, initial where =above, state] (s0) {$q_0$};
	\node[state] (s1) [right of =s0] {$q_1$};
	\node[state] (s2) [left of =s0] {$q_2$};
	
	\path [->]
	(s0) edge [loop below] node [below, sloped] {$\begin{matrix}e_1(a)\\e_2(\epsilon)\end{matrix}$} (s0)
	(s0) edge node [above, sloped] {$e_3(b)$} (s1)
	(s0) edge node [above, sloped] {$e_4(b)$} (s2)
	(s1) edge [loop right] node [above, sloped] {$e_5(b)$} (s1)
	;

        \end{tikzpicture}
		}\hspace{0.5cm}
        \centering
		\subfigure[$\CCa(\Scal_2)$.]{ 
\begin{tikzpicture}[>=stealth',shorten >=1pt,auto,node distance=3.0 cm, scale = 1.0, transform shape,
	>=stealth,inner sep=2pt]
		\node[elliptic state] (s1s2) {$q_1,q_2$};
		\node[elliptic state,initial,initial where=above] (s0s0) [right of =s1s2] {$q_0,q_0$};
		\node[elliptic state] (s2s1) [below of =s0s0] {$q_2,q_1$};
		\node[elliptic state] (s1s1) [right of =s0s0] {$q_1,q_1$};
		\node[elliptic state] (s2s2) [below of =s1s1] {$q_2,q_2$};

		\path [->]
		(s0s0) edge [out = 210, in = 240, loop] node [below, sloped] {$\begin{matrix}(e_1,e_1)\\(e_2,\epsilon)\\(\epsilon,e_2)\end{matrix}$} (s0s0)
		(s0s0) edge node [above, sloped] {$(e_3,e_4)$} (s1s2)
		(s0s0) edge node [above, sloped] {$(e_3,e_3)$} (s1s1)
		(s0s0) edge node [above, sloped] {$(e_4,e_3)$} (s2s1)
		(s0s0) edge node [above, sloped] {$(e_4,e_4)$} (s2s2)
		(s1s1) edge [loop below] node [below, sloped] {$(e_5,e_5)$} (s1s1)
		;
	\end{tikzpicture}
	}
	\caption{LFSA $\Scal_2$ (left) and its self-composition (right, only reachable states illustrated).}
	\label{FA:fig26} 
	\end{figure}
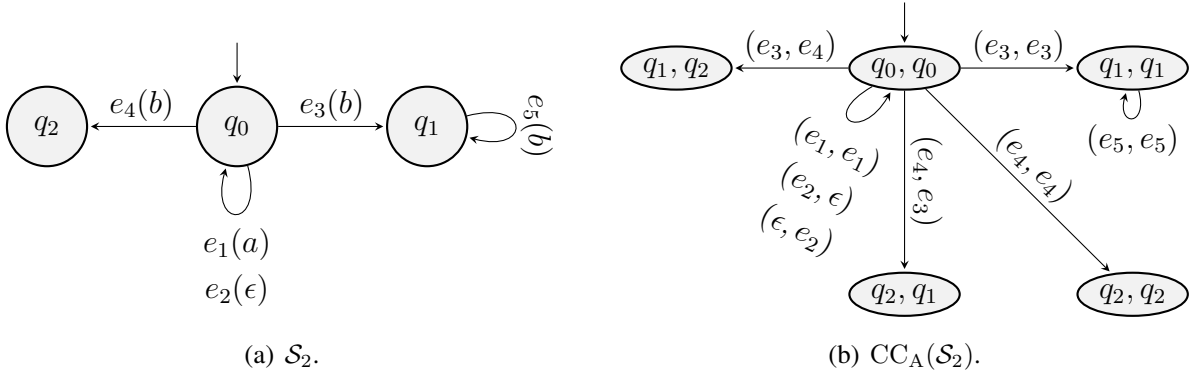 
\end{example}

\section{Detectability}

In this section, we show how to use the concurrent-composition method to verify strong detectability.

The study of the state detection problem dates back to the 1950s \cite{Moore1956} in computer science and
the 1960s \cite{Kalman1963MathDescriptionofLDS} in control science, respectively. In the former,
Moore studied initial-state detection (called Gedanken-experiment) of finite-state machines
which were called Moore machines later; in the latter, Kalman studied initial-state detection (called 
observability) of linear differential equations. The two seminal papers induces many research branches in
computer science and control, e.g., model-based testing of all kinds of reactive systems \cite{Model-BasedTesting2005}
in computer science and observability studies of all kinds of control systems, e.g., arranging from
linear systems \cite{Kalman1963MathDescriptionofLDS,Wonham1985LinearMultiControl}, 
to nonlinear systems \cite{Sontag1979,Conte2007AlgebraicMethodsNonlinearControlSystems,Isidori1999NonConSys}, 
to switched systems \cite{Tanwani2013ObsSwiLinSys}, and also to networked systems
\cite{Kibangou2016ObserNetSys,Angulo2019StrucObserNonNetSys}.

The state detection problem in DESs dates back to the 1980s
\cite{Ramadge1986ObservabilityDES,Ozveren1990ObservabilityDES},
and two widely accepted fundamental notions are \emph{strong detectability}\index{strong detectability} and 
\emph{weak detectability}\index{weak detectability} proposed in 2007 by Shu, Lin, and Ying
\cite{Shu2007Detectability_DES},
where the former implies that there is a delay $k$ such that for \emph{each} event
sequence generated by an LFSA, each prefix of its output sequence of length greater than $k$ allows reconstructing
the current state. The latter relaxes the former by changing \emph{each} to \emph{some}. When long-term behavior
is considered, we let the above conditions apply to all infinite-length generated event sequences (in this case
we call the notions $\omega$-detectability); when short-term behavior is considered, we let them apply to all
finite-length generated event sequences (in this case we call the notions $*$-detectability). 

\begin{definition}[$\omega$-SD \cite{Shu2007Detectability_DES}]\label{FA:def_omega_SD_DetFSA}
	An LFSA $\Scal$ is called $\omega$-\emph{strongly detectable}\index{strong detectability}
	if there exists a positive integer $k$ such that for each infinite-length event sequence $s\in L^{\omega}(\Scal)$,
	$|\Mt(\Scal,\s)|=1$ for every prefix $\s$ of $\ell(s)$ satisfying $|\s|>k$.
\end{definition}

\begin{definition}[$\omega$-WD \cite{Shu2007Detectability_DES}]\label{FA:def_omega_WD_DetFSA}
	An LFSA $\Scal$ is called $\omega$-\emph{weakly detectable}\index{weak detectability}
	if $L^{\omega}(\Scal)\ne\emptyset$ implies there exists an infinite-length event sequence 
	$s\in L^{\omega}(\Scal)$ such that for some positive integer 
	$k$, $|\Mt(\Scal,\s)|=1$ for every prefix $\s$ of $\ell(s)$ satisfying $|\s|> k$.
\end{definition}

\begin{definition}[$*$-SD]\label{FA:def_*SD_DetFSA}
	An LFSA $\Scal$ is called $*$-\emph{strongly detectable}\index{strong detectability}
	if there exists a positive integer $k$ such that for each finite-length event sequence $s\in L(\Scal)$,
	$|\Mt(\Scal,\s)|=1$ for every prefix $\s$ of $\ell(s)$ satisfying $|\s|>k$.
\end{definition}

\begin{definition}[$*$-WD]\label{FA:def_*WD_DetFSA}
	An LFSA $\Scal$ is called $*$-\emph{weakly detectable}\index{weak detectability}
	if there exists a finite-length event sequence $s\in L(\Scal)$ such that for some positive integer 
	$k$, $|\Mt(\Scal,\s)|=1$ for every prefix $\s$ of $\ell(s)$ satisfying $|\s|> k$.
\end{definition}

An exponential-time verification algorithm based on the notion of observer for weak detectability was given in
2007 \cite{Shu2007Detectability_DES}. Recently, verifying weak detectability was proven to be $\PSPACE$-complete
\cite{Zhang2017PSPACEHardnessWeakDetectabilityDES,Masopust2018ComplexityDetectabilityDES}.
In \cite{Shu2011GDetectabilityDES}, a detector method was used to verify strong detectability
in polynomial time, under the two assumptions of deadlock-freeness and divergence-freeness as mentioned above,
where the detector is a simplified version of the observer by splitting the states of the observer into 
subsets of cardinality $2$. We refer the reader to \cite[Remark 2]{Zhang2020DetPNFA} for why the detector method
depends on the two assumptions and without the two assumptions the detector method does not work generally.
In order to verify strong detectability, we choose to characterize its \emph{negation}
(which is essentially different from the way of directly verifying strong detectability adopted in
\cite{Shu2007Detectability_DES,Shu2011GDetectabilityDES}). By definition, the following proposition holds.

	\begin{proposition}[\cite{Zhang2020bookDDS}]\label{FA:prop2}
		An LFSA $\Scal$ is not $\omega$-strongly detectable (resp., $*$-strongly detectable)
		if and only if for every positive integer $k$
		there exists an infinite-length (resp., finite-length) event sequence $s\in L^{\omega}(\Scal)$ 
		(resp., $s\in L(\Scal)$) such that
		$|\Mt(\Scal,\s)|>1$ for some prefix $\s$ of $\ell(s)$ satisfying $|\s|>k$.
	\end{proposition}

With the notion of self-composition of an LFSA $\Scal$, we give sufficient and necessary conditions for
negation of two versions of strong detectability, without any assumption. 

\begin{theorem}[\cite{Zhang2019KDelayStrDetDES,Zhang2020bookDDS}]\label{FA:thm10}
		An LFSA $\Scal$ is not $*$-strongly detectable if and only if in $\CCa(\Scal)$, there exists a run
		\begin{align}\label{FA:eqn6}
			q_0'\xrightarrow[]{s_1'}q_1'\xrightarrow[]{s_2'}q_1'\xrightarrow[]{s_3'}q_2'
		\end{align}
		satisfying
		\begin{align}\label{FA:eqn7}
			&q_0'\in Q_0';\ q_1',q_2'\in Q';\  s_1',s_2',s_3'\in(E')^*;\ \ell(s_2')\in\Sig^+;\ q_2'(L)\ne q_2'(R).
		\end{align}
		An LFSA $\Scal$ is not $\omega$-strongly detectable if and only if in $\CCa(\Scal)$, there exists a run
		\eqref{FA:eqn6} satisfying \eqref{FA:eqn7} and in $\Scal$, there exists a transition cycle reachable from
		$q_2'(L)$.
\end{theorem}

\begin{proof}
	We use Proposition~\ref{FA:prop2} to prove this theorem. We first consider $*$-strong detectability.

	``only if'': Assume $\Scal$ is not $*$-strongly detectable. Then by Proposition~\ref{FA:prop2}, choose 
	$k=|Q|^2$, there exists $s_k\in L(\Scal)$ and $\s\in\Sig^+$ such that  $\Mt(\Scal,\s)>1$, $\s\sqsubset
	\ell(s)$, and $|\s|>k$. Then in $\CCa(\Scal)$, there exists a run $q_0'\xrightarrow[]{s'} q'$ such that
	$q_0'\in Q_0'$, $\ell(s')=\s$ and $q'(L)\ne q'(R)$. Since $|\s|>|Q|^2$ and there exist at most $|Q|^2$ distinct
	states in $\CCa(\Scal)$, by the pigeonhole principle\footnote{If $n$ items are put into $m$ containers,
	with $n>m$, then at least one container must contain more than one item.},
	the run $q_0'\xrightarrow[]{s'} q'$ can be rewritten as
	$q_0'\xrightarrow[]{s_1'} q_1'\xrightarrow[]{s_2'} q_1'\xrightarrow[]{s_3'} q'$, where $\ell(s_2')\in\Sig^+$.

	``if'': Assume in $\CCa(\Scal)$ there exists a run \eqref{FA:eqn6} satisfying \eqref{FA:eqn7}. We choose 
	event sequence $s:=s_1'(L)(s_2'(L))^{k+1}s_3'(L)\in L(\Scal)$, then $|\ell(s)|>k$ and $|\Mt(\Scal,\ell(s))|>1$.
	By Proposition~\ref{FA:prop2}, $\Scal$ is not $*$-strongly detectable.

	We second consider $\omega$-strong detectability. Because a transition cycle reachable from $q_2'(L)$ can
	be repeated arbitrarily often, resulting in an infinite-length run starting from $q_2'(L)$. Then
	based on the above argument for $*$-strong detectability, the sufficient
	and necessary condition for $\omega$-strong detectability also holds.
\end{proof}

\begin{example}
	Reconsider the LFSA $\Scal_2$ in Figure~\ref{FA:fig26} (left) and its self-composition $\CCa(\Scal_2)$
	in Figure~\ref{FA:fig26} (right). In $\CCa(\Scal_2)$, one sees a run 
	$$(q_0,q_0)\xrightarrow[]{(e_1,e_1)}(q_0,q_0)\xrightarrow[]{(e_3,e_4)}(q_1,q_2)$$ such that
	$\ell((e_1,e_1))=a$ is of positive length and $q_1\ne q_2$. Then by Theorem~\ref{FA:thm10}, $\Scal_2$ 
	is not $*$-strongly detectable. In addition, in $\Scal_2$, there is a self-loop on $q_1$, hence also
	by Theorem~\ref{FA:thm10}, $\Scal_2$ is not $\omega$-strongly detectable.
\end{example}

%

\section{Diagnosability}

In order to define \emph{diagnosability}\index{diagnosability} for an LFSA $\Scal=(Q,E,\dt,Q_0,\Sig,\ell)$, 
we specify a subset $\Ef\subset E$ of faulty events. Diagnosability describes whether one can use an observed
output sequence to determine whether some faulty event has occurred. For an event sequence $s\in E^*$, 
$\Ef\in s$ denotes that some element of $\Ef$ appears in $s$.

\begin{definition}[Diag \cite{Sampath1995DiagnosabilityDES}]\label{FA:def11}
	Consider an LFSA $\Scal=(Q,E,\dt,Q_0,\Sig,\ell)$ and a subset ${\Ef}\subset E$ of faulty events.
		$\Scal$ is called \emph{$\Ef$-diagnosable}\index{diagnosability} if 
		\begin{align*}
			&(\exists k\in\N)(\forall s\in L(\Scal)\cap E^*{\Ef})(\forall s':ss'\in L(\Scal))\\
			&[(|s'|> k)\implies {\bf D}],
		\end{align*}
		where ${\bf D}=(\forall s''\in L(\Scal))[(\ell(s'')=\ell(ss')) \implies ({\Ef}\in s'')]$.
\end{definition}

Intuitively, if $\Scal$ is $\Ef$-diagnosable, then once a faulty event (e.g., the last event in $s$)
occurs, one can make sure that some faulty event has occurred after at least $k$ subsequent events (e.g., 
$s'$) occur by observing output sequences.

In 1995, Sampath et al. \cite{Sampath1995DiagnosabilityDES} proposed a
diagnoser method to verify diagnosability. The diagnoser of an LFSA records state estimates
along observed output sequences and also records fault propagation along transitions of states of the 
LFSA. The same as the observer mentioned above, the diagnoser also has exponential complexity, and diagnosability
is verifiable by a relatively simple cycle 
condition on the diagnoser. Hence diagnosability can be verified in exponential time. Also the same as 
the case that the observer method depends on the two assumptions of deadlock-freeness and divergence-freeness
when being applied to verify detectability \cite{Shu2007Detectability_DES}, the diagnoser 
method also depends on the two assumptions when being applied to verify diagnosability. Later in 2001, a twin-plant
method with polynomial complexity was proposed by Jiang et al. \cite{Jiang2001PolyAlgorithmDiagnosabilityDES} to verify
diagnosability in polynomial time. Because in a twin plant, only observable transitions are synchronized, the method
also depends on the two assumptions. One year later, Yoo and Lafortune 
\cite{Yoo2002DiagnosabiliyDESPTime} proposed a verifier method to verify diagnosability 
in polynomial time, where in a verifier, observable transitions are synchronized, unobservable transitions
are also considered but their events' positions (left or right) are neglected, so that the method also depends on the
two assumptions. From then on, in many papers, verification of all kinds of variants of diagnosability
depends on the two assumptions. The two assumptions were removed by Cassez and Tripakis
\cite{Cassez2008FaultDiagnosisStDyObser} in 2008
by using a generalized version of the twin-plant structure to verify \emph{negation} of diagnosability
in polynomial time, where in the generalized version of the twin plant, observable transitions are synchronized,
unobservable transitions are also considered but their events' positions (left or right) are also considered. 
The generalized version of the twin-plant structure and the concurrent-composition structure \cite{Zhang2020DetPNFA}
were proposed in a similar way: they were proposed by characterizing \emph{negation} of inference-based 
properties.

In order to verify $E_f$-diagnosability of $\Scal$, we use the concurrent composition $\CCa(\Scal_{\fsf},
\Scal_{\nsf})$ (similar to but simpler than the generalized version of the twin plant proposed in 
\cite{Cassez2008FaultDiagnosisStDyObser})
of the faulty subautomaton $\Scal_{\fsf}$ and the normal subautomaton $\Scal_{\nsf}$, where 
$\Scal_{\fsf}$ is obtained from $\Scal$ by only keeping faulty transitions and all their predecessors and 
successors, $\Scal_{\nsf}$ is obtained from $\Scal$ by removing all faulty transitions. $\CCa(\Scal_{\fsf},
\Scal_{\nsf})$ is computed similarly as in Definition~\ref{FA:def9}.

\begin{theorem}[\cite{Zhang2021UnifyingDetDiagPred}]\label{FA:thm13}  
	Consider an LFSA $\Scal=(Q,E,\dt,Q_0,\Sig,\ell)$ and a subset ${\Ef}\subset E$ of faulty events.
	$\Scal$ is not ${\Ef}$-diagnosable if and only if in $\CCa(\Scal_{\fsf},\Scal_{\nsf})$, there exists a run
	\begin{align}\label{FA:eqn9} 
		q_0'\xrightarrow[]{s_1'}q_1'\xrightarrow[]{e'}q_2'\xrightarrow[]{s_2'}q_3'\xrightarrow[]{s_3'}q_3'
	\end{align}
	satisfying
	\begin{align}\label{FA:eqn10}
		&q_0'\in Q_0';\ e'(L)\in {\Ef};\ |s_3'(L)|>0.
	\end{align}
\end{theorem}

\begin{proof}
	By definition, $\Scal$ is not $\Ef$-diagnosable if and only if
	\begin{align*}
		&(\forall k\in\N)(\exists s_k\in L(\Scal)\cap E^*{\Ef})(\exists s_k':s_ks_k'\in L(\Scal))(\exists s_k''\in L(\Scal))\\
		&[(|s_k'|> k) \wedge (\ell(s_k'')=\ell(s_ks_k')) \wedge ({\Ef}\notin s_k'')].
	\end{align*}
	Choose sufficiently large $k$, by the finiteness of the number of states of $\Scal$ and the pigeonhole principle,
	$\Scal$ is not $\Ef$-diagnosable if and only if in $\CCa(\Scal_{\fsf},\Scal_{\nsf})$, there exists a run 
	\eqref{FA:eqn9} satisfying \eqref{FA:eqn10}.
\end{proof}

\begin{example}\label{FA:exam22} 
	Consider the LFSA $\Scal_3$ in Figure~\ref{FA:fig28}. We compute part of the concurrent composition
	$\CCa(\Scal_{3\fsf},\Scal_{3\nsf})$ as in Figure~\ref{FA:fig29}. In $\CCa(\Scal_{3\fsf},\Scal_{3\nsf})$, 
	one sees a run $(q_0,q_0)\xrightarrow[]{(e_1,e_1)}(q_1,q_2)\xrightarrow[]{(e_2,e_2)}(q_3,q_4)\xrightarrow[] 
	{(\f,\ep)}(q_5,q_4)\xrightarrow[]{(u,\ep)}(q_5,q_4)$, which satisfies \eqref{FA:eqn10}. Then by 
	Theorem~\ref{FA:thm13}, $\Scal_3$ is not $\fs$-diagnosable.
	\begin{figure}[!htpb]
			\tikzset{global scale/.style={
    scale=#1,
    every node/.append style={scale=#1}}}
		\begin{center}
			\begin{tikzpicture}[global scale = 1.0,
				>=stealth',shorten >=1pt,thick,auto,node distance=2.5 cm, scale = 1.0, transform shape,
	->,>=stealth,inner sep=2pt,
				every transition/.style={draw=red,fill=red,minimum width=1mm,minimum height=3.5mm},
				every place/.style={draw=blue,fill=blue!20,minimum size=7mm}]
				\tikzstyle{emptynode}=[inner sep=0,outer sep=0]
				\node[state, initial, initial where = left] (x0) {$q_0$};
				\node[state] (x1) [right of = x0] {$q_1$};
				\node[state] (x2) [below of = x1] {$q_2$};
				\node[state] (x3) [right of = x1] {$q_3$};
				\node[state] (x4) [below of = x3] {$q_4$};
				\node[state] (x5) [right of = x3] {$q_5$};

				\path[->]
				(x0) edge node [above, sloped] {$e_1(a)$} (x1)
				(x0) edge node [above, sloped] {$e_1(a)$} (x2)
				(x1) edge node [above, sloped] {$e_2(b)$} (x3)
				(x2) edge node [above, sloped] {$e_2(b)$} (x4)
				(x3) edge node [above, sloped] {${\f}(\ep)$} (x5)
				(x5) edge [loop right] node {$u(\ep)$} (x5)
				(x4) edge [loop right] node {$u(\ep)$} (x4)
				;
			\end{tikzpicture}
	\end{center}
	\caption{LFSA $\Scal_3$, where only event $\f$ is faulty.}
	\label{FA:fig28} 
	\end{figure}
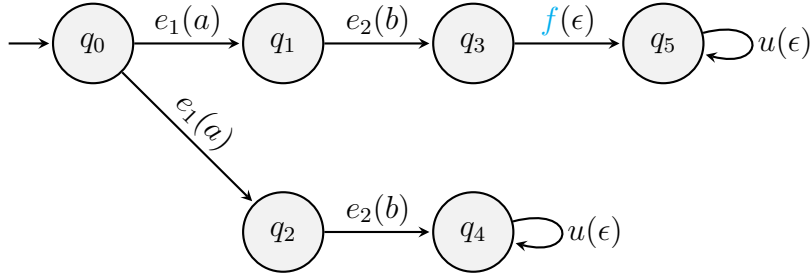
	\begin{figure}[!htpb]
			\tikzset{global scale/.style={
    scale=#1,
    every node/.append style={scale=#1}}}
		\begin{center}
			\begin{tikzpicture}[global scale = 1.0,
				>=stealth',shorten >=1pt,thick,auto,node distance=3.2 cm, scale = 1.0, transform shape,
	->,>=stealth,inner sep=2pt,
				every transition/.style={draw=red,fill=red,minimum width=1mm,minimum height=3.5mm},
				every place/.style={draw=blue,fill=blue!20,minimum size=7mm}]
				\tikzstyle{emptynode}=[inner sep=0,outer sep=0]
				\node[elliptic state, initial, initial where = left] (00) {$q_0,q_0$};
				\node[elliptic state] (12) [right of = 00] {$q_1,q_2$};
				\node[elliptic state] (34) [right of = 12] {$q_3,q_4$};
				\node[elliptic state] (54) [right of = 34] {$q_5,q_4$};

				\path[->]
				(00) edge node {$(e_1,e_1)$} (12)
				(12) edge node {$(e_2,e_2)$} (34)
				(34) edge [loop above] node {$(\ep,u)$} (34)
				(34) edge node {$(\f,\ep)$} (54)
				(54) edge [loop right] node {$\begin{matrix}(\ep,u)\\(u,\ep)\end{matrix}$} (54)
				;
			\end{tikzpicture}
	\end{center}
	\caption{Part of $\CCa(\Scal_{3\fsf},\Scal_{3\nsf})$, where $\Scal_3$ is shown in Figure~\ref{FA:fig28}.}
	\label{FA:fig29} 
	\end{figure}
\end{example}

\section{Predictability}

Differently from diagnosability, \emph{predictability}\index{predictability} describes whether one can use an 
observed output sequence to make sure some faulty event will be certain to occur.

\begin{definition}[Pred \cite{Genc2009PredictabilityDES}]\label{FA:def12} 
	Consider an LFSA $\Scal=(Q,E,\dt,Q_0,\Sig,\ell)$ and a subset ${\Ef}\subset E$ of faulty events.
		$\Scal$ is called \emph{$\Ef$-predictable}\index{predictability} if 
			\begin{align*}
				&(\exists k\in\N)(\forall s\in L(\Scal)\cap E^*{\Ef})(\exists s'\sqsubset s:{\Ef}\notin s')
				(\forall uv\in L(\Scal))\\
				&[((\ell(s')=\ell(u)) \wedge ({\Ef}\notin u) \wedge (|v|> k)) \implies ({\Ef}\in v)].
			\end{align*}	
	\end{definition}

Intuitively, if $\Scal$ is $\Ef$-predictable, then once a faulty event will definitely occur,
then before any faulty event occurs, one
can make sure that after a common time delay (representing the number of occurrences of events, e.g., $k$),
all generated event sequences with the same observation
without any faulty event must be continued by an event sequence containing a faulty event, so as to raise an alarm
to definite occurrence of some faulty event.

In order to verify $\Ef$-predictability of $\Scal$, we use the self-composition $\CCa(\Scal_{\nsf},
\Scal_{\nsf})$ of the normal subautomaton $\Scal_{\nsf}$. $\CCa(\Scal_{\nsf},
\Scal_{\nsf})$ is also computed similarly as in Definition~\ref{FA:def9}.

	\begin{theorem}[\cite{Zhang2021UnifyingDetDiagPred}]\label{FA:thm14}  
		Consider an LFSA $\Scal=(Q,E,\dt,Q_0,\Sig,\ell)$ and a subset ${\Ef}\subset E$ of faulty events.
		$\Scal$ is not $\Ef$-predictable if and only if in $\CCa(\Scal_{\nsf},\Scal_{\nsf})$, there exists a run
		\begin{align}\label{FA:eqn11}
			q_0'\xrightarrow[]{s_1'}q_1'
		\end{align}
		such that
		\begin{subequations}\label{FA:eqn12}
			\begin{align}
				&q_0'\in Q_0';\\
				&(q_1'(L),\ef,q)\in\dt\text{ for some }\ef\in {\Ef}\text{ and  }q\in Q;\\
				&\text{ in }\Scal_{\nsf},\text{ there is a transition cycle
				reachable from }q_1'(R).
			\end{align}
		\end{subequations}
\end{theorem}

\begin{proof}
	By definition, $\Scal$ is not $\Ef$-predictable if and only if
	\begin{align*}
		&(\forall k\in\N)(\exists s_k\in L(\Scal)\cap E^*{\Ef})(\forall s_k'\sqsubset s_k:{\Ef}\notin s_k')
		(\exists u_kv_k\in L(\Scal))\\
		&[(\ell(s_k')=\ell(u_k)) \wedge ({\Ef}\notin u_kv_k) \wedge (|v_k|> k) ].
	\end{align*}
	Choose sufficiently large $k$, by the finiteness of the number of states of $\Scal$ and the pigeonhole principle,
	$\Scal$ is not $\Ef$-predictable if and only if in $\CCa(\Scal_{\nsf},\Scal_{\nsf})$, there exists a run 
	\eqref{FA:eqn11} satisfying \eqref{FA:eqn12}.
\end{proof}

\begin{example}\label{FA:exam23}
	Reconsider the LFSA $\Scal_3$ in Figure~\ref{FA:fig28}. Part of the concurrent composition
	$\CCa(\Scal_{3\nsf},\Scal_{3\nsf})$ can be obtained from 
	Figure~\ref{FA:fig29} by removing the transition with event $({\f},\ep)$.
	Then in $\CCa(\Scal_{3\nsf},\Scal_{3\nsf})$, 
	one sees a run $(q_0,q_0)\xrightarrow[]{(e_1,e_1)}(q_1,q_2)\xrightarrow[]{(e_2,e_2)}(q_3,q_4)$,
	which satisfies \eqref{FA:eqn12}. Then by Theorem~\ref{FA:thm14}, $\Scal_3$ is not $\fs$-predictable.
\end{example}

\section{Standard opacity}
\label{sec:StanOpa}

Opacity is a concealment-based (confidentiality) property which was first proposed by Mazar\'{e}
\cite{Mazare2004Opacity} in 2004.
Opacity describes whether the visit of \emph{secrets} in a system
could be forbidden from being leaked to an external \emph{intruder}, given that the intruder knows complete 
knowledge of the system's structure but can only see generated outputs.
It has been widely used to describe all kinds of scenarios in cyber security/privacy problems
such as the dinning cryptographers problem \cite{Chaum1988DiningCryptographers},
encryption using pseudo random generators and tracking of mobile agents in sensor networks 
\cite{Saboori2010PhDThesisOpacity}, ensuring privacy in location-based services \cite{Wu2014PhDThesisOpacity},
the indoor location privacy problem using obfuscation 
\cite{Wu2018SyntheObfuscation,Wu2016ObfuscatorSynthesis,Goes20181IndoorPrivacyObsfucation}.

In \cite{Bryans2008OpacityTransitionSystems}, a general run-based opacity framework
was proposed for labeled transition systems (which contain LFSAs, labeled Petri nets, etc., as subclasses),
where such a system is opaque if for every secret run,
there exists a non-secret run such that the two runs produce the same observation.
Later on, two special types of secrets are studied: subsets of event sequences (aka traces)
and subsets of states. According to the two types of secrets, opacity is classified into language-based opacity
and state-based opacity. The former refers to for every secret generated trace, there is a non-secret
generated trace such that they produce the same observation; the latter refers to whenever a run passes through a 
secret state at some instant, there exists another run that does not pass any secret state at the same 
instant such that the two runs produce the same observation. Language-based opacity is more involved, because 
it is already undecidable for LFSAs which contain no observable events \cite{Bryans2008OpacityTransitionSystems};
particularly, when secret languages and non-secret languages are regular, language-based opacity is decidable 
in exponential time \cite{Lin2011OpacityDES}. State-based opacity is relatively simpler.
When the time instant of visiting secret states is specified as
the initial time, the current time, any past time, and at most $K$ steps prior to the current time,
the notions of state-based opacity can be formulated as \emph{initial-state opacity} (ISO)
\cite{Saboori2013InitialStateOpacity},
\emph{current-state opacity} (CSO) \cite{Cassez2009DynamicOpcity}, 
\emph{infinite-step opacity}  (InfSO)
\cite{Saboori2012InfiniteStepOpacity},
and \emph{$K$-step opacity} ($K$SO) \cite{Saboori2009KStepOpacityFA}, respectively. The problems of verifying
the four types of state-based opacity are $\PSPACE$-complete in LFSAs \cite{Saboori2013InitialStateOpacity,Cassez2009DynamicOpcity,Saboori2012InfiniteStepOpacity},
the four properties and the special case of language-based opacity studied in \cite{Lin2011OpacityDES}
are polynomially reducible to each other 
\cite{Wu2013ComparativeAnalysisOpacity,Balun2021ComparingOpacityinDES}.


Next, we show a concurrent-composition method to verify the four properties of state-based opacity. One can directly
use an observer to verify CSO \cite{Cassez2009DynamicOpcity,Saboori2007CurrentStateOpacity} and
directly use a reverse observer 
to verify ISO \cite{Wu2013ComparativeAnalysisOpacity}. The verification methods in 
\cite{Cassez2009DynamicOpcity,Saboori2007CurrentStateOpacity,Wu2013ComparativeAnalysisOpacity} are currently the most efficient methods for
verifying CSO and ISO. However, verifying InfSO and $K$SO are more difficult, currently one cannot see any 
possibility of directly using an observer and a reverse observer to do their verification. The
concurrent-composition method
to be shown to verify InfSO and $K$SO is more efficient than the initial-state-estimator method 
\cite{Saboori2012InfiniteStepOpacity} and the two-way-observer method (i.e., the alternating product of an observer
and a reverse observer) \cite{Yin2017TWObserverInfiniteStepOpacity}.

\begin{definition}[ISO \cite{Saboori2013InitialStateOpacity}]\label{FA:def_ISO}
	Consider an LFSA $\Scal=(Q,E,\dt,Q_0,\Sig,\ell)$ and a subset $\QS\subset Q$ of secret states.
	$\Scal$ is called \emph{initial-state opaque} (ISO) \emph{with respect to $\QS$}
	if for every run $q_0\xrightarrow[]{s}q$ with 
	$q_0\in Q_0\cap \QS$, there exists a run $q_0'\xrightarrow[]{s'}q'$ such that $q_0'\in Q_0\setminus \QS$
	and $\ell(s)=\ell(s')$.
\end{definition}

From now on, ISO is short for ``initial-state opacity'' or ``initial-state opaque'' adapted to the 
context. Analogous for CSO, InfSO, and $K$SO.

Intuitively, if an LFSA is ISO, then an external intruder cannot make sure whether the initial state
is secret by observing generated label sequences.

\begin{definition}[CSO \cite{Cassez2009DynamicOpcity}]\label{FA:def_CSO}
	Consider an LFSA $\Scal=(Q,E,\dt,Q_0,\Sig,\ell)$ and a subset $\QS\subset Q$ of secret states.
	$\Scal$ is called \emph{current-state opaque} (CSO)  \emph{with respect to $\QS$}
	if for every run $q_0\xrightarrow[]{s}q$ with 
	$q_0\in Q_0$ and $q\in \QS$, there exists a run $q_0'\xrightarrow[]{s'}q'$ such that $q_0'\in Q_0$,
	$q'\in Q\setminus \QS$, and $\ell(s)=\ell(s')$.
\end{definition}

If an LFSA is CSO, then an external intruder cannot make sure whether the current state is
secret by observing generated label sequences.

\begin{definition}[InfSO \cite{Saboori2012InfiniteStepOpacity}]\label{FA:def_InfSO}
	Consider an LFSA $\Scal=(Q,E,\dt,Q_0,\Sig,\ell)$ and a subset $\QS\subset Q$ of secret states.
	$\Scal$ is called \emph{infinite-step opaque} (InfSO)  \emph{with respect to $\QS$}
	if for every run $q_0\xrightarrow[]{s_1}q_1\xrightarrow
	[]{s_2}q_2$ with $q_0\in Q_0$ and $q_1\in \QS$, there exists a run $q_0'\xrightarrow[]{s_1'}q_1'\xrightarrow
	[]{s_2'}q_2'$ such that $q_0'\in Q_0$, $q_1'\in Q\setminus \QS$, $\ell(s_1)=\ell(s_1')$, and $\ell(s_2)=\ell(s_2')$.
\end{definition}

\begin{definition}[$K$SO \cite{Saboori2009KStepOpacityFA}]\label{FA:def_KSO}
	Consider an LFSA $\Scal=(Q,E,\dt,Q_0,\Sig,\ell)$, a subset $\QS\subset Q$ of secret states, and a positive
	integer $K$. $\Scal$ is called \emph{$K$-step opaque} ($K$SO)  \emph{with respect to $\QS$}
	if for every run $q_0\xrightarrow[]{s_1}q_1\xrightarrow
	[]{s_2}q_2$ with $q_0\in Q_0$, $q_1\in \QS$, and $|\ell(s_2)|\le K$, there exists a run
	$q_0'\xrightarrow[]{s_1'}q_1'\xrightarrow
	[]{s_2'}q_2'$ such that $q_0'\in Q_0$, $q_1'\in Q\setminus \QS$, $\ell(s_1)=\ell(s_1')$, and $\ell(s_2)=\ell(s_2')$.
\end{definition}

If an LFSA is InfSO ($K$SO), then an external intruder cannot make sure whether 
any past state (at most $K$ steps prior to the current time) is secret by observing generated label sequences.

In order to verify CSO, the notion of observer is enough. Observer is the classical powerset construction
used for determinizing nondeterministic finite automata with $\ep$-transitions.

\begin{definition}[\cite{Sipser2006TheoryofComputation}]\label{FA:def10}
	Consider an LFSA $\Scal=(Q,E,\dt,Q_0,\Sig,\ell)$. Its \emph{observer} $\Scal_{\obs}$ (the term ``observer''
		was used in \cite{Shu2007Detectability_DES} and hereafter) is defined by
		a deterministic finite automaton $$(Q_{\obs},\ell(E_o),\dt_{\obs},q_{0\obs}),$$ where
		\begin{enumerate}
			\item $Q_{\obs}=2^Q$,
			\item $\ell(E_o)=\ell(E)\setminus\{\ep\}$,
			\item for all $x\in Q_{\obs}$ and $a\in\ell(E_o)$, $\dt_{\obs}(x,a)= \bigcup_{q\in x}\bigcup_{
				\substack{e_a\in E_o\\\ell(e_a)=a\\ s\in (E_{uo})^*}}\dt(q,e_as)$,
			\item $q_{0\obs}=\bigcup_{q_0\in Q_0}\bigcup_{s\in (E_{uo})^*}\dt(q_0,s)$ (i.e., $\UR(Q_0)$).
		\end{enumerate}
	\end{definition}

By definition, for all $a\in\ell(E_o)$, one has $\dt_{\obs}(\emptyset,a)=\emptyset$. The size of $\Scal_{\obs}$
is $O(2^{|Q|}|\ell(E_o)|)$, the time consumption of computing $\Scal_{\obs}$ is $O(2^{|Q|}|Q|^2|\ell(E_o)||E|)$:
for every subset $x\subset Q$ and every label $a\in\ell(E_o)$, the time consumption of computing $\dt(x,a)$
is bounded (from above) by $|Q|^2|E_o|+|Q|^2|E_{uo}|=|Q|^2|E|$.

\begin{example}\label{FA:exam20} 
	Part of the observer $\Scal_{2\obs}$ of the LFSA $\Scal_2$ in Figure~\ref{FA:fig26} is shown in
	Figure~\ref{FA:fig27}.
	\begin{figure}[!htbp]
	\centering
	\begin{tikzpicture}[>=stealth',shorten >=1pt,auto,node distance=2.5 cm, scale = 1.0, transform shape,
	>=stealth,inner sep=2pt]
		\node[elliptic state,initial,initial where=left] (s0) {$\{q_0\}$};
		\node[elliptic state] (s1s2) [right of =s0] {$\{q_1,q_2\}$};
		\node[elliptic state] (s1) [right of =s1s2] {$\{q_1\}$};
		\node[elliptic state] (emptys) [right of =s1] {$\emptyset$};
		\node[elliptic state] (s2) [right of =emptys] {$\{q_2\}$};

		\path [->]
		(s0) edge [loop above] node {$a$} (s0)
		(s0) edge node [above, sloped] {$b$} (s1s2)
		(s1s2) edge node [above, sloped] {$b$} (s1)
		(s1) edge [loop above] node {$b$} (s1)
		(s1s2) edge [bend right] node [below, sloped] {$a$} (emptys)
		(s1) edge node {$a$} (emptys)
		(emptys) edge [loop above] node {$a,b$} (emptys)
		(s2) edge node [above, sloped] {$a,b$} (emptys)
		;
	\end{tikzpicture}
	\caption{Part of observer $\Scal_{2\obs}$ of the LFSA $\Scal_2$ in Figure~\ref{FA:fig26}.}
	\label{FA:fig27} 
	\end{figure} 
\end{example}

\begin{theorem}[\cite{Cassez2009DynamicOpcity,Saboori2007CurrentStateOpacity}]\label{FA:thm_CSO}
	An LFSA $\Scal=(Q,E,\dt,Q_0,\Sig,\ell)$ is CSO with respect to $\QS\subset Q$ if and only if for every
	nonempty state $x$ reachable in $\Scal_{\obs}$, $x\not\subset \QS$.
\end{theorem}

\begin{example}\label{FA:exam30}
	Reconsider the LFSA $\Scal_2$ in Figure~\ref{FA:fig26} and its observer $\Scal_{2\obs}$ in
	Figure~\ref{FA:fig27}. By Theorem~\ref{FA:thm_CSO}, $\Scal_2$ is CSO with respect to $\red\{q_2\}$.
\end{example}

We use the concurrent composition $\CCa(\Scal_{\vep},\Scal_{\obs}^{\vep})$ to verify the other three notions
of opacity, where $\Scal_{\vep}$ is obtained from $\Scal$ by changing each transition $q\xrightarrow[]{e}q'$
to $q\xrightarrow[]{\ell(e)}q'$ if $e\in E_o$, to $q\xrightarrow[]{\vep}q'$ if $e\in E_{uo}$, and replacing
the labeling function of $\Scal$ by the map $\ell'$ on $\ell(E_o)\cup\{\vep\}$ satisfying that $\ell'|_{\ell(E_o)}$
(the restriction of $\ell'$ to $\ell(E_o)$) is the identity map and $\ell'(\vep)=\ep$; 
$\Scal_{\obs}^{\vep}$ is obtained from $\Scal_{\obs}$ by adding an additional event $\vep$ and the labeling 
function $\ell'$ of $\Scal_{\vep}$. In this particular case, $\CCa(\Scal_{\vep},\Scal_{\obs}^{\vep})$ is almost
the same as the parallel composition of $\Scal_{\vep}$ and $\Scal_{\obs}^{\vep}$ in 
\cite[Page 80]{Cassandras2009DESbook}.

The size of $\CCa(\Scal_{\vep},\Scal_{\obs}^{\vep})$ is $O(|Q|2^{|Q|}(|\ell(E_o)||Q|+|E_{uo}||Q|))=
O(|Q|^22^{|Q|}|E|)$. The time consumption of computing $\CCa(\Scal_{\vep},\Scal_{\obs}^{\vep})$ is also
$O(|Q|^22^{|Q|}|E|)$ after $\Scal_{\vep}$ and $\Scal_{\obs}^{\vep}$ have been computed.

\begin{theorem}\label{FA:thm_ISO}
	An LFSA $\Scal=(Q,E,\dt,Q_0,\Sig,\ell)$ is ISO with respect to $\QS\subset Q$ if and only if $Q_0\ne\emptyset
	\implies Q_0\not\subset \QS$ and for every
	$q_0\in Q_0\cap \QS$, in concurrent composition $\CCa(\Scal_{\vep},\Scal_{\obs}^{\vep})$, all states reachable 
	from $(q_0,\UR(Q_0\setminus \QS))$ are of the form $(-,x)$ with $x\ne\emptyset$.
\end{theorem}

\begin{theorem}[\cite{Balun2021KSO_ConCurren}]\label{FA:thm_InfSO}
	An LFSA $\Scal=(Q,E,\dt,Q_0,\Sig,\ell)$ is InfSO with respect to $\QS\subset Q$ if and only if for every
	nonempty state $x$ reachable in $\Scal_{\obs}$, one has $x\not\subset \QS$ and for every
	$q\in x\cap \QS$, in concurrent composition $\CCa(\Scal_{\vep},\Scal_{\obs}^{\vep})$, all states reachable
	from $(q,x\setminus \QS)$ are of the form $(-,x')$ with $x'\ne\emptyset$.
\end{theorem}

\begin{theorem}[\cite{Balun2021KSO_ConCurren}]\label{FA:thm_KSO}
	An LFSA $\Scal=(Q,E,\dt,Q_0,\Sig,\ell)$ is $K$SO with respect to $\QS\subset Q$ 
	if and only if for every nonempty state $x$ reachable in $\Scal_{\obs}$, one has $x\not\subset \QS$ and for every
	$q\in x\cap \QS$, in concurrent composition $\CCa(\Scal_{\vep},\Scal_{\obs}^{\vep})$, for every run
	$(q,x\setminus \QS)\xrightarrow[]{s'}(q',x')$ with $|\ell(s')|\le K$, $x'\ne\emptyset$.
\end{theorem}

Theorems~\ref{FA:thm_CSO}, \ref{FA:thm_ISO}, \ref{FA:thm_InfSO}, and \ref{FA:thm_KSO} directly follow from
definition. By definition, one directly sees the following corollaries.
If an LFSA $\Scal$ is $K$SO (with respect to ${\red Q_S}\subset Q$), then it is $K'$SO for any $K'< K$.
Conversely, if $\Scal$ is not $K$SO with $K>2^{|Q|}-2$, then it is not $K'$SO for some $K'\le 2^{|Q|}-2$,
because $\Scal_{\obs}$ has at most $2^{|Q|}-1$ nonempty states; then it is not $(2^{|Q|}-2)$SO.
Hence the verification of $K$SO based on Theorem~\ref{FA:thm_KSO}
does not depend on $K$ if $K>2^{|Q|}-2$. The verification algorithms shown in Theorems~\ref{FA:thm_CSO},
\ref{FA:thm_ISO},
\ref{FA:thm_InfSO}, and \ref{FA:thm_KSO} all run in time $O(2^{|Q|}|Q|^2|\ell(E_o)||E|)$.
The upper bound $2^{|Q|}-2$ for $K$ was obtained in \cite{Yin2017TWObserverInfiniteStepOpacity}.
The upper bound for $K$ obtained in \cite{Saboori2009KStepOpacityFA} is $2^{|Q|^2}-2$.
Compared with the concurrent-composition method, the relative inefficiency of the two-way observer method
\cite{Yin2017TWObserverInfiniteStepOpacity} comes from computing a reverse observer (with the same
complexity as computing an observer) and the alternating product (i.e., the so-called two-way observer) of the
observer and the reverse observer. The verification
algorithms obtained in \cite{Saboori2012InfiniteStepOpacity,Saboori2009KStepOpacityFA} have even higher
complexity.

\begin{corollary}[\cite{Balun2021KSO_ConCurren}]\label{FA:thm_KSO'}
	An LFSA $\Scal=(Q,E,\dt,Q_0,\Sig,\ell)$ is $K$SO with respect to $\QS\subset Q$ 
	if and only if it is $\min\{K,|2^{|Q|}-2\}$SO with respect to $\QS$.
\end{corollary}

\begin{corollary}\label{FA:thm_InfSO_KSO}
	An LFSA $\Scal=(Q,E,\dt,Q_0,\Sig,\ell)$ is InfSO with respect to $\QS\subset Q$ if and only if it is
	$K$SO with respect to $\QS$ with $K>2^{|Q|}-2$.
\end{corollary}

\begin{example}\label{FA:exam31}
	Reconsider the LFSA $\Scal_2$ in Figure~\ref{FA:fig26} and its observer $\Scal_{2\obs}$ in
	Figure~\ref{FA:fig27}. The corresponding $\Scal_{2\vep}$ is shown in 
	Figure~\ref{FA:fig36}. The concurrent composition $\CCa(\Scal_{2\vep},\Scal_{2\obs}^{\vep})$ is shown in 
	Figure~\ref{FA:fig37}.
	\begin{figure}[H]
     \centering
	 	\begin{tikzpicture}[>=stealth',shorten >=1pt,auto,node distance=2.5 cm, scale = 1.0, transform shape,
	>=stealth,inner sep=2pt]

	\node[initial, initial where =above, state] (s0) {$q_0$};
	\node[state] (s1) [right of =s0] {$q_1$};
	\node[state] (s2) [left of =s0] {$q_2$};
	
	\path [->]
	(s0) edge [loop below] node [below, sloped] {$a,\vep$} (s0)
	(s0) edge node [above, sloped] {$b$} (s1)
	(s0) edge node [above, sloped] {$b$} (s2)
	(s1) edge [loop right] node {$b$} (s1)
	;

        \end{tikzpicture}
		
	\caption{LFSA $\Scal_{2\vep}$ corresponding to the LFSA $\Scal_2$ in Figure~\ref{FA:fig26}.}
	\label{FA:fig36} 
	\end{figure}
	\begin{figure}[H]
		\begin{center}
			\begin{tikzpicture}[>=stealth',shorten >=1pt,auto,node distance=3.6 cm, scale = 1.0, transform shape,
	>=stealth,inner sep=2pt]
		\node[elliptic state] (2-12) {$q_2,\{q_1,q_2\}$};
		\node[elliptic state,initial,initial where=above] (0-0) [right of =2-12] {$q_0,\{q_0\}$};
		\node[elliptic state] (1-12) [right of =0-0] {$q_1,\{q_1,q_2\}$};
		\node[elliptic state] (1-1) [right of =1-12] {$q_1,\{q_1\}$};
		\node[elliptic state] (2-1) [below of =2-12] {$q_2,\{q_1\}$};
		\node[elliptic state] (1-2) [right of =2-1] {$q_1,\{q_2\}$};
		\node[elliptic state] (1-emptys) [right of =1-2] {$q_1,\emptyset$};

		\path [->]
		(0-0) edge node [above, sloped] {$(b,b)$} (2-12)
		(0-0) edge node [above, sloped] {$(b,b)$} (1-12)
		(0-0) edge [loop below] node {$(a,a),(\vep,\ep)$} (0-0)
		(1-12) edge node [above, sloped] {$(b,b)$} (1-1)
		(1-1) edge [loop right] node {$(b,b)$} (1-1)
		(1-2) edge node [above, sloped] {$(b,b)$} (1-emptys)
		(1-emptys) edge [loop right] node {$(b,b)$} (1-emptys)
		;
	\end{tikzpicture}
		\end{center}
		\caption{Part of $\CCa(\Scal_{2\vep},\Scal_{2\obs}^{\vep})$ corresponding to the LFSA $\Scal_2$ in Figure~\ref{FA:fig26}.}
		\label{FA:fig37}
	\end{figure}

	By Theorem~\ref{FA:thm_ISO}, $\Scal_2$ is not ISO with respect to $\red \{q_0\}$, because $\red q_0$ is the unique
	initial state. In addition,
	one has $\Scal_2$ is InfSO with respect to $\red\{q_2\}$ by Theorem~\ref{FA:thm_InfSO},
	because the unique reachable state of $\Scal_{2\obs}$ containing $\red q_2$ is $\{q_1,{\red q_2}\}$ and
	in $\CCa(\Scal_{2\vep},\Scal_{2\obs}^{\vep})$, there is no state reachable from $({\red q_2},\{q_1\})$ of 
	the form $(-,\emptyset)$. By the reachable state $\red\{q_1\}$ of $\Scal_{2\obs}$, one sees $\Scal_2$ is not InfSO 
	with respect to $\red\{q_1\}$, which can also be seen from the fact that $\{{\red q_1},q_2\}$ is reachable 
	in $\Scal_{2\obs}$ and in $\CCa(\Scal_{2\vep},\Scal_{2\obs}^{\vep})$,
	the state $({\red q_1},\emptyset)$ is reachable from $({\red q_1},\{q_2\})$.
\end{example}

\section{Strong opacity}

In Section~\ref{sec:StanOpa}, variants of notions of opacity were shown to describe the ability of an 
LFSA to forbid its visit of secret states from being leaked to an external intruder. Sometimes, such 
``standard'' opacity is not sufficiently strong, e.g., in some CSO LFSA, when observing a generated 
label sequence, one can make sure that some secret state must have been visited, although cannot make sure
of the exact visit instant of time. Consider the following LFSA $\Scal_4$:
\begin{figure}[H]
		\begin{center}
			\begin{tikzpicture}[>=stealth',shorten >=1pt,auto,node distance=3.3 cm, scale = 1.0, transform shape,
	>=stealth,inner sep=2pt]
		\node[state,initial, initial where = left] (1) {$q_1$};
		\node[state] (2) [right of =1] {$\red q_2$};
		\node[state, initial, initial where = left] (3) [right of =2] {$\red q_3$};
		\node[state] (4) [right of =3] {$q_4$};

		\path [->]
		(1) edge node [above, sloped] {$a$} (2)
		(3) edge node [above, sloped] {$a$} (4)
		;
	\end{tikzpicture}
		\end{center}
		\caption{LSFA $\Scal_4$, where $\ell(a)=a$, $q_2$ and $q_3$ are secret, $q_1$ and $q_4$ are not.}
		\label{FA:fig38}
\end{figure}
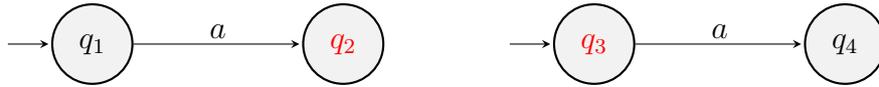
Automaton $\Scal_4$ is CSO with respect to $\{q_2,q_3\}$. When observing $a$, one can make sure that 
at least one secret state has been visited, in detail, if $q_1\xrightarrow[]{a}q_2$ was generated then
$q_2$ was visited, if $q_3\xrightarrow[]{a}q_4$ was generated then $q_3$ was visited. 
This leads to a ``strong version'' of CSO which guarantees that
an intruder cannot make sure whether the current state is secret, and can also guarantee that 
the intruder cannot make sure whether some secret state has been 
visited. Analogously, the other three standard versions of opacity studied in Section~\ref{sec:StanOpa} could
also be reformulated as their strong versions.

In order to define strong versions of state-based opacity, we define a \emph{non-secret} run of an LFSA by
a run containing no secret states.

\begin{definition}[SISO \cite{Han2021StrongOpacityDES}]\label{FA:def_SISO}
	Consider an LFSA $\Scal=(Q,E,\dt,Q_0,\Sig,\ell)$ and a subset $\QS\subset Q$ of secret states.
	$\Scal$ is called \emph{strongly initial-state opaque} (SISO) \emph{with respect to $\QS$}
	if for every run $q_0\xrightarrow[]{s}q$ with $q_0\in Q_0\cap \QS$, there exists a \emph{non-secret} run
	$q_0'\xrightarrow[]{s'}q'$ such that $q_0'\in Q_0\setminus \QS$
	and $\ell(s)=\ell(s')$.
\end{definition}

If an LFSA is SISO, then an external intruder cannot make sure whether the initial state
is secret and cannot make sure whether some secret state has been visited either, by observing generated
label sequences.

\begin{definition}[SCSO \cite{Han2021StrongOpacityDES}]\label{FA:def_SCSO}
	Consider an LFSA $\Scal=(Q,E,\dt,Q_0,\Sig,\ell)$ and a subset $\QS\subset Q$ of secret states.
	$\Scal$ is called \emph{strongly current-state opaque} (SCSO)  \emph{with respect to $\QS$}
	if for every run $q_0\xrightarrow[]{s}q$ with 
	$q_0\in Q_0$ and $q\in \QS$, there exists a \emph{non-secret} run $q_0'\xrightarrow[]{s'}q'$ 
	such that $q_0'\in Q_0$ and $\ell(s)=\ell(s')$.
\end{definition}

If an LFSA is SCSO, then an external intruder cannot make sure whether the current state is
secret and cannot make sure whether some secret state has been visited either, by observing generated label sequences.

\begin{definition}[SInfSO \cite{Ma2021StrongInfSODES}]\label{FA:def_SInfSO}
	Consider an LFSA $\Scal=(Q,E,\dt,Q_0,\Sig,\ell)$ and a subset $\QS\subset Q$ of secret states.
	$\Scal$ is called \emph{strongly infinite-step opaque} (SInfSO)  \emph{with respect to $\QS$}
	if for every run $q_0\xrightarrow[]{s_1}q_1\xrightarrow
	[]{s_2}q_2$ with $q_0\in Q_0$ and $q_1\in \QS$, there exists a \emph{non-secret} run $q_0'\xrightarrow[]{s_1'}q_1'\xrightarrow
	[]{s_2'}q_2'$ such that $q_0'\in Q_0$, $\ell(s_1)=\ell(s_1')$, and $\ell(s_2)=\ell(s_2')$.
\end{definition}

\begin{definition}[S$K$SO]\label{FA:def_SKSO}
	Consider an LFSA $\Scal=(Q,E,\dt,Q_0,\Sig,\ell)$, a subset $\QS\subset Q$ of secret states, and a positive
	integer $K$. $\Scal$ is called \emph{strongly $K$-step opaque}\footnote{Note that the current S$K$SO is slightly
	stronger than the $K$-step strong opacity proposed in \cite{Falcone2015StrongKStepOpacity}, where in the latter,
	$q_0'\xrightarrow[]{s_1'}q_1'\xrightarrow[]{s_2'}q_2'$ is not necessarily non-secret, but only $q_1'\xrightarrow
	[]{s_2'}q_2'$ is necessarily non-secret.}
	(S$K$SO)  \emph{with respect to $\QS$}
	if for every run $q_0\xrightarrow[]{s_1}q_1\xrightarrow
	[]{s_2}q_2$ with $q_0\in Q_0$, $q_1\in \QS$, and $|\ell(s_2)|\le K$, there exists a \emph{non-secret} run
	$q_0'\xrightarrow[]{s_1'}q_1'\xrightarrow
	[]{s_2'}q_2'$ such that $q_0'\in Q_0$, $\ell(s_1)=\ell(s_1')$, and $\ell(s_2)=\ell(s_2')$.
\end{definition}

From now on, SISO is short for ``strong initial-state opacity'' or ``strongly initial-state opaque'' adapted to the 
context. Analogous for SCSO, SInfSO, and S$K$SO.

If an LFSA is SInfSO (S$K$SO), then an external intruder cannot make sure whether 
any past state (at most $K$ steps prior to the current time) is secret and cannot make sure whether some secret 
state has been visited either, by observing generated label sequences.

Next we use the concurrent-composition structure to do verification for the four strong versions of state-based
opacity, where the derived verification algorithms are more efficient than the $K$/Inf-step recognizer method
proposed in \cite{Ma2021StrongInfSODES}.

Consider an LFSA $\Scal=(Q,E,\dt,Q_0,\Sig,\ell)$ and a subset $\QS\subset Q$ of secret states, let $$\Scal_{\dss}=
(Q_{\dss},E_{\dss},\dt_{\dss},Q_{0\dss},\Sig,\ell_{\dss})$$
be the accessible part of the remainder of $\Scal$ by \emph{deleting secret states} (dss) of $\Scal$.
Let $$\Scal_{\dss\obs}=(Q_{\dss\obs},\ell_{\dss}(E_{\dss})\setminus\{\ep\},\dt_{\dss\obs},q_{0\dss\obs})$$ 
be the observer of $\Scal_{\dss}$.

Similarly to $\Scal_{\obs}$, the size of $\Scal_{\dss\obs}$ is $O(2^{|Q|}|\ell(E_o)|)$, the time consumption of
computing $\Scal_{\dss\obs}$ is $O(2^{|Q|}|Q|^2|\ell(E_o)||E|)$. The size of $\Scal_{\dss\obs}$ is slightly 
smaller than that of $\Scal_{\obs}$.

We will use the concurrent composition $\CCa(\Scal_{\vep},\Scal_{\dss\obs}^{\vep})$ to verify the four strong
versions of opacity. Unlike CSO and ISO, SCSO and SISO cannot be verified by directly using the notions of
observer and reverse observer. 

Similarly to $\CCa(\Scal_{\vep},\Scal_{\obs}^{\vep})$, the size of $\CCa(\Scal_{\vep},\Scal_{\dss\obs}^{\vep})$ 
is $O(|Q|^22^{|Q|}|E|)$. The time consumption of computing $\CCa(\Scal_{\vep},\Scal_{\dss\obs}^{\vep})$ is also
$O(|Q|^22^{|Q|}|E|)$ after $\Scal_{\vep}$ and $\Scal_{\dss\obs}^{\vep}$ have been computed.

\begin{theorem}[\cite{Han2021StrongOpacityDES}]\label{FA:thm_SCSO}
	An LFSA $\Scal=(Q,E,\dt,Q_0,\Sig,\ell)$ is SCSO with respect to $\QS\subset Q$ if and only if for every
	state $(q,x)$ reachable in $\CCa(\Scal_{\vep},\Scal_{\dss\obs}^{\vep})$, if $q\in \QS$ then $x\ne\emptyset$.
\end{theorem}

\begin{theorem}[\cite{Han2021StrongOpacityDES}]\label{FA:thm_SISO}
	An LFSA $\Scal=(Q,E,\dt,Q_0,\Sig,\ell)$ is SISO with respect to $\QS\subset Q$ if and only if $Q_0\ne\emptyset
	\implies Q_0\not\subset \QS$ and for every
	$q_0\in Q_0\cap \QS$, in concurrent composition $\CCa(\Scal_{\vep},\Scal_{\dss\obs}^{\vep})$, all states reachable 
	from $(q_0,q_{0\dss\obs})$ are of the form $(-,x)$ with $x\ne\emptyset$.
\end{theorem}

\begin{theorem}[\cite{Han2021StrongOpacityDES}]\label{FA:thm_SInfSO}
	An LFSA $\Scal=(Q,E,\dt,Q_0,\Sig,\ell)$ is SInfSO with respect to $\QS\subset Q$, if and only if, (\romannumeral1)
	for every state $(q,x)$ reachable in concurrent composition $\CCa(\Scal_{\vep},\Scal_{\dss\obs}^{\vep})$,
	if $q\in \QS$ then $x\ne\emptyset$ and all states reachable from $(q,x)$ are of the form $(-,x')$ with
	$x'\ne\emptyset$, if and only if, (\romannumeral2) all states $(q'',x'')$ reachable in $\CCa(\Scal_{\vep},
	\Scal_{\dss\obs}^{\vep})$ satisfy $x''\ne\emptyset$.
\end{theorem}

\begin{proof}
	By definition, (\romannumeral1) is equivalent for $\Scal$ to be SInfSO with respect to $\QS$. 

	(\romannumeral2) $\implies$ (\romannumeral1): This trivially holds.

	(\romannumeral1) $\implies$ (\romannumeral2): Consider a state $(q'',x'')$ reachable in $\CCa(\Scal_{\vep},
	\Scal_{\dss\obs}^{\vep})$. If there is a run from some initial state of $\CCa(\Scal_{\vep},
	\Scal_{\dss\obs}^{\vep})$ to $(q'',x'')$ containing a state $(q,x)$ with $q\in \QS$, then by (\romannumeral1),
	one has $x''\ne\emptyset$; otherwise one also has $x''\ne\emptyset$ because $q''\in x''$ by definition of
	$\CCa(\Scal_{\vep},\Scal_{\dss\obs}^{\vep})$.
\end{proof}

\begin{theorem}\label{FA:thm_SKSO}
	An LFSA $\Scal=(Q,E,\dt,Q_0,\Sig,\ell)$ is S$K$SO with respect to $\QS\subset Q$ 
	if and only if for every
	state $(q,x)$ reachable in concurrent composition $\CCa(\Scal_{\vep},\Scal_{\dss\obs}^{\vep})$, if $q\in \QS$
	then $x\ne\emptyset$ and for every run
	$(q,x)\xrightarrow[]{s'}(q',x')$ with $|\ell(s')|\le K$, $x'\ne\emptyset$.
\end{theorem}

Similarly to the standard versions of opacity, by definition, one also directly sees the following
corollaries, because $\Scal_{\dss\obs}$ has at most $2^{|Q\setminus Q_S|}-1$ nonempty states.
Hence the verification of S$K$SO based on Theorem~\ref{FA:thm_SKSO} does not depend on $K$ if $K>2^{|Q\setminus Q_S|}-2$.
The verification algorithms shown in Theorems~\ref{FA:thm_SCSO}, \ref{FA:thm_SISO},
\ref{FA:thm_SInfSO}, and \ref{FA:thm_SKSO} all run in time $O(2^{|Q|}|Q|^2|\ell(E_o)||E|)$.

\begin{corollary}\label{FA:thm_SKSO'}
	An LFSA $\Scal=(Q,E,\dt,Q_0,\Sig,\ell)$ is S$K$SO with respect to $\QS\subset Q$ 
	if and only if it is S$\min\{K,2^{|Q\setminus Q_S|}-2\}$SO with respect to $\QS$.
\end{corollary}

\begin{corollary}\label{FA:thm_SInfSO_SKSO}
	An LFSA $\Scal=(Q,E,\dt,Q_0,\Sig,\ell)$ is SInfSO with respect to $\QS\subset Q$ if and only if it is
	S$K$SO with respect to $\QS$ and positive integer $K$ with $K>2^{|Q\setminus Q_S|}-2$.
\end{corollary}

\begin{remark}
	The main time consumption in the verification algorithms shown in Theorems~\ref{FA:thm_CSO}, \ref{FA:thm_ISO},
	\ref{FA:thm_InfSO}, \ref{FA:thm_KSO}, \ref{FA:thm_SCSO}, \ref{FA:thm_SISO}, \ref{FA:thm_SInfSO}, and
	\ref{FA:thm_SKSO} comes from computing the corresponding observer. If the observer is not explicitly 
	computed, then by using nondeterministic search, verification can be done in $\PSPACE$.
\end{remark}

\begin{example}\label{FA:exam32}
	Consider the following LFSA $\Scal_5$:
	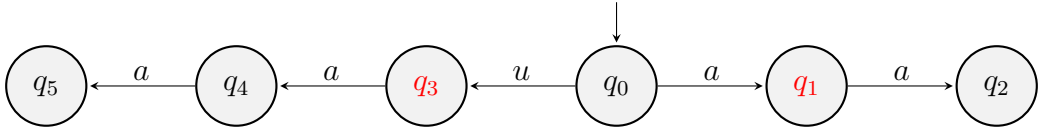
\begin{figure}[H]
		\begin{center}
			\begin{tikzpicture}[>=stealth',shorten >=1pt,auto,node distance=2.5 cm, scale = 1.0, transform shape,
	>=stealth,inner sep=2pt]
		\node[state,initial, initial where = above] (0) {$q_0$};
		\node[state] (1) [right of =0] {$\red q_1$};
		\node[state] (2) [right of =1] {$q_2$};
		\node[state] (3) [left of =0] {$\red q_3$};
		\node[state] (4) [left of =3] {$q_4$};
		\node[state] (5) [left of =4] {$q_5$};

		\path [->]
		(0) edge node [above, sloped] {$a$} (1)
		(1) edge node [above, sloped] {$a$} (2)
		(0) edge node [above, sloped] {$u$} (3)
		(3) edge node [above, sloped] {$a$} (4)
		(4) edge node [above, sloped] {$a$} (5)
		;
	\end{tikzpicture}
		\end{center}
		\caption{LSFA $\Scal_5$, where $\ell(u)=\ep$, $\ell(a)=a$, $q_1$ and $q_3$ are secret, the other 
		states are not.}
		\label{FA:fig39} 
	\end{figure}
	We verify whether $\Scal_5$ is InfSO or SInfSO with respect to $\red\{q_1,q_3\}$ by Theorem~\ref{FA:thm_InfSO}
	and Theorem~\ref{FA:thm_SInfSO}. By Theorem~\ref{FA:thm_InfSO}, we compute $\Scal_{5\vep}$, $\Scal_{5\obs}$,
	and $\CCa(\Scal_{5\vep},\Scal_{5\obs}^{\vep})$ as follows:
	\begin{figure}[H]
		\begin{center}
			\begin{tikzpicture}[>=stealth',shorten >=1pt,auto,node distance=2.5 cm, scale = 1.0, transform shape,
	>=stealth,inner sep=2pt]
		\node[state,initial, initial where = above] (0) {$q_0$};
		\node[state] (1) [right of =0] {$\red q_1$};
		\node[state] (2) [right of =1] {$q_2$};
		\node[state] (3) [left of =0] {$\red q_3$};
		\node[state] (4) [left of =3] {$q_4$};
		\node[state] (5) [left of =4] {$q_5$};

		\path [->]
		(0) edge node [above, sloped] {$a$} (1)
		(1) edge node [above, sloped] {$a$} (2)
		(0) edge node [above, sloped] {$\vep$} (3)
		(3) edge node [above, sloped] {$a$} (4)
		(4) edge node [above, sloped] {$a$} (5)
		;
	\end{tikzpicture}
		\end{center}
		\caption{$\Scal_{5\vep}$ corresponding to the LFSA $\Scal_5$ in Figure~\ref{FA:fig39}.}
		\label{FA:fig40} 
	\end{figure}
	\begin{figure}[H]
		\begin{center}
			\begin{tikzpicture}[>=stealth',shorten >=1pt,auto,node distance=2.5 cm, scale = 1.0, transform shape,
	>=stealth,inner sep=2pt]
		\node[elliptic state,initial, initial where = left] (03) {$\{q_0,{\red q_3}\}$};
		\node[elliptic state] (14) [right of =03] {$\{{\red q_1},q_4\}$};
		\node[elliptic state] (25) [right of =14] {$\{q_2,q_5\}$};
		\node[elliptic state] (emptys) [right of =25] {$\emptyset$};

		\node[elliptic state] (0) [above of =03] {$\{q_0\}$};
		\node[elliptic state] (1) [right of =0] {$\{{\red q_1}\}$};
		\node[elliptic state] (2) [right of =1] {$\{q_2\}$};

		\node[elliptic state] (3) [below of =03] {$\{{\red q_3}\}$};
		\node[elliptic state] (4) [right of =3] {$\{q_4\}$};
		\node[elliptic state] (5) [right of =4] {$\{q_5\}$};

		\path [->]
		(03) edge node [above, sloped] {$a$} (14)
		(14) edge node [above, sloped] {$a$} (25)
		(25) edge node [above, sloped] {$a$} (emptys)
		(emptys) edge [loop right] node {$a$} (emptys)

		(0) edge node [above, sloped] {$a$} (1)
		(1) edge node [above, sloped] {$a$} (2)
		(3) edge node [above, sloped] {$a$} (4)
		(4) edge node [above, sloped] {$a$} (5)
		(2) edge [bend left] node [above, sloped] {$a$} (emptys)
		(5) edge [bend right] node [above, sloped] {$a$} (emptys)
		;
	\end{tikzpicture}
		\end{center}
		\caption{Part of $\Scal_{5\obs}$ corresponding to the LFSA $\Scal_5$ in Figure~\ref{FA:fig39}.}
		\label{FA:fig41} 
	\end{figure}
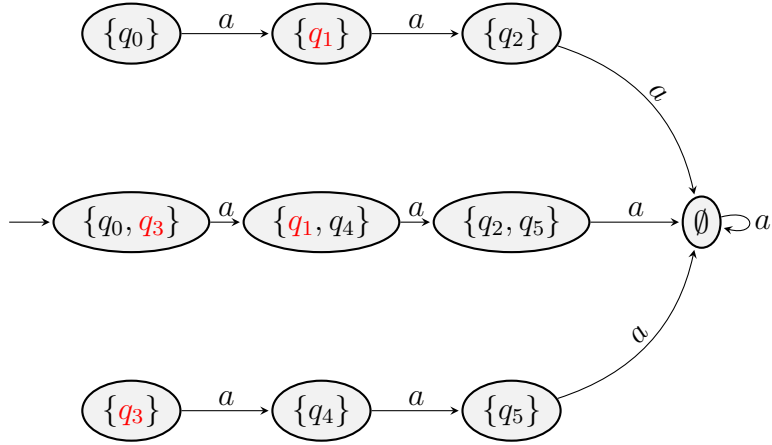
	\begin{figure}[H]
		\begin{center}
			\begin{tikzpicture}[>=stealth',shorten >=1pt,auto,node distance=3.3 cm, scale = 1.0, transform shape,
	>=stealth,inner sep=2pt]
		\node[elliptic state] (3-0) {${\red q_3},\{q_0\}$};
		\node[elliptic state] (4-1) [right of =3-0] {$q_4,\{{\red q_1}\}$};
		\node[elliptic state] (5-2) [right of =4-1] {$q_5,\{q_2\}$};
		\node[elliptic state] (1-4) [right of =5-2] {${\red q_1},\{q_4\}$};
		\node[elliptic state] (2-5) [right of =1-4] {$q_2,\{q_5\}$};

		\path [->]
		(3-0) edge node [above, sloped] {$(a,a)$} (4-1)
		(4-1) edge node [above, sloped] {$(a,a)$} (5-2)
		(1-4) edge node {$(a,a)$} (2-5)
		;
		\end{tikzpicture}		
		\end{center}
		\caption{Part of $\CCa(\Scal_{5\vep},\Scal_{5\obs}^{\vep})$ corresponding to the LFSA $\Scal_5$ 
		in Figure~\ref{FA:fig39}.}
		\label{FA:fig42} 
	\end{figure}
	In observer $\Scal_{5\obs}$, the reachable states containing secret states are $\{q_0,{\red q_3}\}$ and $\{{\red 
	q_1},q_4\}$.
	In $\CCa(\Scal_{5\vep},\Scal_{5\obs}^{\vep})$, the states reachable from $({\red q_3},\{q_0\})$ and $({\red q_1},\{q_4\})$
	all satisfy that their right components are nonempty. Then by Theorem~\ref{FA:thm_InfSO}, $\Scal_5$ is InfSO
	with respect to $\red\{q_1,q_3\}$.

	By Theorem~\ref{FA:thm_SInfSO}, we compute $\Scal_{5\dss}$, $\Scal_{5\dss\obs}$,
	and $\CCa(\Scal_{5\vep},\Scal_{5\dss\obs}^{\vep})$ as follows:
	\begin{figure}[H]
     \centering
	 \subfigure[$\Scal_{5\dss}$.]{
\begin{tikzpicture}[>=stealth',shorten >=1pt,auto,node distance=3.0 cm, scale = 1.0, transform shape,
	>=stealth,inner sep=2pt]

		\node[state,initial,initial where=left] (0) {$q_0$};

        \end{tikzpicture}
		}\hspace{0.5cm}
        \centering
		\subfigure[$\Scal_{5\dss\obs}$.]{ 
\begin{tikzpicture}[>=stealth',shorten >=1pt,auto,node distance=3.0 cm, scale = 1.0, transform shape,
	>=stealth,inner sep=2pt]
		
		\node[elliptic state,initial,initial where=left] (0obs) {$\{q_0\}$};
		\node[elliptic state, right of =0obs] (emptys) {$\emptyset$};

		\path [->]
		(0obs) edge node [above, sloped] {$a$} (emptys)
		(emptys) edge [loop right] node {$a$} (emptys)
		;
	\end{tikzpicture}
	}
	\caption{$\Scal_{5\dss}$ and $\Scal_{5\dss\obs}$ corresponding to the LFSA $\Scal_5$ 
	in Figure~\ref{FA:fig39}.}
	\label{FA:fig43} 
	\end{figure}
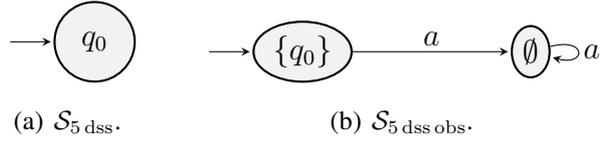 
	\begin{figure}[H]
		\begin{center}
			\begin{tikzpicture}[>=stealth',shorten >=1pt,auto,node distance=3.0 cm, scale = 1.0, transform shape,
	>=stealth,inner sep=2pt]
		\node[elliptic state, initial, initial where = left] (0-0) {$q_0,\{q_0\}$};
		\node[elliptic state] (1-emp) [right of =0-0] {${\red q_1},\emptyset$};
		\node[elliptic state] (2-emp) [right of =1-emp] {$q_2,\emptyset$};
		\node[elliptic state] (3-0) [below of =1-emp] {${\red q_3},\{q_0\}$};
		\node[elliptic state] (4-emp) [right of =3-0] {$q_4,\emptyset$};
		\node[elliptic state] (5-emp) [right of =4-emp] {$q_5,\emptyset$};

		\path [->]
		(0-0) edge node [above, sloped] {$(a,a)$} (1-emp)
		(1-emp) edge node [above, sloped] {$(a,a)$} (2-emp)
		(0-0) edge node [above, sloped] {$(\vep,\ep)$} (3-0)
		(3-0) edge node [above, sloped] {$(a,a)$} (4-emp)
		(4-emp) edge node [above, sloped] {$(a,a)$} (5-emp)
		;
		\end{tikzpicture}		
		\end{center}
		\caption{Part of $\CCa(\Scal_{5\vep},\Scal_{5\dss\obs}^{\vep})$ (all reachable states illustrated)
		corresponding to the LFSA $\Scal_5$ 
		in Figure~\ref{FA:fig39}.}
		\label{FA:fig44} 
	\end{figure}
	In $\CCa(\Scal_{5\vep},\Scal_{5\dss\obs}^{\vep})$, there exist reachable states whose right components 
	are equal to $\emptyset$, then by Theorem~\ref{FA:thm_SInfSO}, $\Scal_5$ is not SInfSO with respect to
	$\red\{q_1,q_3\}$.
\end{example}

\section{Conclusion}

In this paper, a unified concurrent-composition method was given to verify inference-based properties and 
concealment-based properties in labeled finite-state automata. Compared with the previous verification algorithms 
in the literature, the concurrent-composition method does not depend on assumptions and is more efficient.
These results for the first time showed that many inference-based properties and concealment-based properties 
can be unified into one mathematical framework, although the two categories of properties look quite different.
This similarity between the two categories has never been revealed before. It is interesting to explore other 
usages of the concurrent-composition method, e.g., what other properties could be verified by the  method, what
other kinds of models in discrete-event systems could be dealt with by the method, and what other problems (e.g.,
enforcement) can be solved by the method.

\section*{Appendix}

We briefly review the twin-plant method proposed in \cite{Jiang2001PolyAlgorithmDiagnosabilityDES} and the verifier
method proposed in \cite{Yoo2002DiagnosabiliyDESPTime} used for verifying diagnosability and show that
they usually do not work without the two assumptions of liveness/deadlock-freeness and divergence-freeness.
We also briefly show the coincident similarity between the concurrent composition and the generalized version
of the twin plant proposed in \cite{Cassez2008FaultDiagnosisStDyObser}.

For brevity, we consider an LFSA $\Scal=(Q,E,\dt,Q_0,\Sig,\ell)$ in which $\ell|_{E_o}$ is the identity mapping
and $Q_0$ is a singleton and denoted by $\{q_0\}$. Recall that $E_o$ is 
the set of observable events, and $E_{uo}=E\setminus E_o$ is the set of unobservable events. Consider a 
single faulty event $\f\in E_{uo}$.

The twin plant $\TwPl_{\Scal}$ of $\Scal$ proposed in \cite{Jiang2001PolyAlgorithmDiagnosabilityDES} is constructed as 
follows:
\begin{enumerate}[(1)]
	\item Construct the automaton $\Scal_{\phi}=(Q_{\phi},E_o,(x_0,\phi),\dt_{\phi})$, where
		the initial state is $(x_0,\phi)\in Q_{\phi}$, $Q_{\phi}=Q\times \{\phi,F\}$; for all $(x_1,\phi),(x_2,l_2)
		\in Q_{\phi}$ and $t\in E_o$, $((x_1,\phi),t,(x_2,l_2))\in \dt_{\phi}$ if and only if there is a run 
		$x_1\xrightarrow[]{st}x_2$ in $\Scal$ such that $s\in(E_{uo})^*$, and $l_2=F$ if and only
		if $\f$ appears in at least one such $s$; for all $(x_1,F),(x_2,l_2)
		\in Q_{\phi}$ and $t\in E_o$, $((x_1,F),t,(x_2,l_2))\in \dt_{\phi}$ if and only if there is a run 
		$x_1\xrightarrow[]{st}x_2$ in $\Scal$ such that $s\in E_{uo}^*$ and $l_2=F$\footnote{Here
		$F$ denotes propagation of $\f$, i.e., along every run of $\Scal_{\phi}$, once a state has its right 
		component equal to $F$, then all subsequent states have their right components equal to $F$.}.
	\item The twin plant $\TwPl_{\Scal}$ is the parallel composition $\Scal_{\phi}||\Scal_{\phi}$ of
		$\Scal_{\phi}$ with itself, where the 
		parallel composition is as in \cite[Page 80]{Cassandras2009DESbook}. In this special case, 
		$\Scal_{\phi}||\Scal_{\phi}$ is almost the same 
		as the self-composition $\CCa(\Scal_{\phi})$ because $\Scal_{\phi}$ contains no unobservable events.
		After replacing each event $e_o$ in $\Scal_{\phi}||\Scal_{\phi}$ by $(e_o,e_o)$, $\CCa(\Scal_{\phi})$
		is obtained.
\end{enumerate}
\begin{proposition}[\cite{Jiang2001PolyAlgorithmDiagnosabilityDES}]\label{prop1_responses}
	A live and divergence-free $\Scal$ is $\fs$-diagnosable if and only if in $\TwPl_{\Scal}$ all states of 
	all cycles are of the form $((-,l_1),(-,l_2))$ with $l_1= l_2$.
\end{proposition}

\begin{example}
	Proposition~\ref{prop1_responses} does not generally hold for $\Scal$ that is not
	live or divergence-free. Consider the following LFSA $\Scal_6$:
		\begin{figure}[H]
		\tikzset{global scale/.style={
    scale=#1,
    every node/.append style={scale=#1}}}
		\begin{center}
			\begin{tikzpicture}[global scale = 1.0,
				>=stealth',shorten >=1pt,thick,auto,node distance=3.0 cm, scale = 1.0, transform shape,
	->,>=stealth,inner sep=2pt,
				every transition/.style={draw=red,fill=red,minimum width=1mm,minimum height=3.5mm},
				every place/.style={draw=blue,fill=blue!20,minimum size=7mm}]
				\tikzstyle{emptynode}=[inner sep=0,outer sep=0]
				\node[state, initial, initial where = above] (x0) {$x_0$};
				\node[state] (x1) [left of = x0] {$x_1$};
				\node[state] (x2) [right of = x0] {$x_2$};

				\path[->]
				(x0) edge node [above, sloped] {$\f$} (x1)
				(x0) edge node [above, sloped] {$u$} (x2)
				(x1) edge [loop left] node {$u$} (x1)
				(x2) edge [loop right] node {$u$} (x2)
				;
			\end{tikzpicture}
			\caption{LFSA $\Scal_6$, where $\ell({\f})=\ell(u)=\ep$.}
			\label{fig1_response}
		\end{center}
	\end{figure}
	By definition, $\Scal_{6\phi}$ consists of only the initial state $(x_0,\phi)$. Hence 
	by Proposition~\ref{prop1_responses}, $\Scal_6$ is $\fs$-diagnosable vacuously. However, by definition,
	$\Scal_6$ is not $\fs$-diagnosable.
\end{example}

The verifier $\Ver_{\Scal}=(Q_{\Ver},E,(x_0,N,x_0,N),\dt_{\Ver})$ of $\Scal$ proposed in 
\cite{Yoo2002DiagnosabiliyDESPTime} is constructed as follows: $Q_{\Ver}=Q\times\{N,F\}\times Q\times\{N,F\}$,
for all $(x_1,l_1,x_2,l_2)\in Q_{\Ver}$, $\s_o\in E_o$, and $\s_{uo}\in E_{uo}\setminus
\fs$,
\begin{enumerate}[(i)]
	\item $((x_1,l_1,x_2,l_2),\f,(x_1',F,x_2,l_2))\in\dt_{\Ver}$ if and only if $(x_1,\f,x_1')\in\dt$,
	\item $((x_1,l_1,x_2,l_2),\f,(x_1,l_1,x_2',F))\in\dt_{\Ver}$ if and only if $(x_2,\f,x_2')\in\dt$,
	\item $((x_1,l_1,x_2,l_2),\f,(x_1',F,x_2',F))\in\dt_{\Ver}$ if and only if $(x_1,\f,x_1'),(x_2,\f,x_2')\in\dt$,
	\item $((x_1,l_1,x_2,l_2),\s_{uo},(x_1',l_1,x_2,l_2))\in\dt_{\Ver}$ if and only if $(x_1,\s_{uo},x_1')\in\dt$,
	\item $((x_1,l_1,x_2,l_2),\s_{uo},(x_1,l_1,x_2',l_2))\in\dt_{\Ver}$ if and only if $(x_2,\s_{uo},x_2')\in\dt$,
	\item $((x_1,l_1,x_2,l_2),\s_{uo},(x_1',l_1,x_2',l_2))\in\dt_{\Ver}$ if and only if $(x_1,\s_{uo},x_1'),(x_2,\s_{uo},x_2')\in\dt$,
	\item $((x_1,l_1,x_2,l_2),\s_{o},(x_1',l_1,x_2',l_2))\in\dt_{\Ver}$ if and only if $(x_1,\s_{o},x_1'),(x_2,\s_o,x_2')\in\dt$.
\end{enumerate}

\begin{proposition}[\cite{Yoo2002DiagnosabiliyDESPTime}]\label{prop2_responses}
	A live and divergence-free $\Scal$ is $\fs$-diagnosable if and only if in $\Ver_{\Scal}$ all states of all
	cycles are of the form $(-,l_1,-,l_2)$ with $l_1 = l_2$\footnote{In $\Ver_{\Scal}$, 
	if there is a cycle containing a state of the form $(-,l_1,-,l_2)$ with $l_1\ne l_2$, then
	either (1) all states in the cycle are of the form
	$(-,F,-,N)$ or (2) all states in the cycle are of the form $(-,N,-,F)$.}.
\end{proposition}

\begin{example}
	Proposition~\ref{prop2_responses} does not generally hold for $\Scal$ that is not
	live or divergence-free. Consider the following LFSA $\Scal_7$:
	\begin{figure}[H]
		\tikzset{global scale/.style={
    scale=#1,
    every node/.append style={scale=#1}}}
		\begin{center}
			\begin{tikzpicture}[global scale = 1.0,
				>=stealth',shorten >=1pt,thick,auto,node distance=3.0 cm, scale = 1.0, transform shape,
	->,>=stealth,inner sep=2pt,
				every transition/.style={draw=red,fill=red,minimum width=1mm,minimum height=3.5mm},
				every place/.style={draw=blue,fill=blue!20,minimum size=7mm}]
				\tikzstyle{emptynode}=[inner sep=0,outer sep=0]
				\node[state, initial, initial where = above] (x0) {$x_0$};
				\node[state] (x1) [left of = x0] {$x_1$};
				\node[state] (x2) [right of = x0] {$x_2$};

				\path[->]
				(x0) edge node [above, sloped] {$\f$} (x1)
				(x0) edge node [above, sloped] {$u$} (x2)
				(x2) edge [loop right] node {$u$} (x2)
				;
			\end{tikzpicture}
			\caption{LFSA $\Scal_7$, where $\ell(\f)=\ell(u)=\ep$.}
			\label{fig2_response}
		\end{center}
	\end{figure}
	Part of $\Ver_{\Scal_7}$ is shown as follows:
	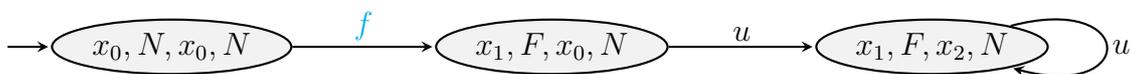
\begin{figure}[H]
		\tikzset{global scale/.style={
    scale=#1,
    every node/.append style={scale=#1}}}
		\begin{center}
			\begin{tikzpicture}[global scale = 1.0,
				>=stealth',shorten >=1pt,thick,auto,node distance=5.0 cm, scale = 1.0, transform shape,
	->,>=stealth,inner sep=2pt,
				every transition/.style={draw=red,fill=red,minimum width=1mm,minimum height=3.5mm},
				every place/.style={draw=blue,fill=blue!20,minimum size=7mm}]
				\tikzstyle{emptynode}=[inner sep=0,outer sep=0]
				\node[elliptic state, initial, initial where = left] (0N0N) {$x_0,N,x_0,N$};
				\node[elliptic state] (1F0N) [right of = 0N0N] {$x_1,F,x_0,N$};
				\node[elliptic state] (1F2N) [right of = 1F0N] {$x_1,F,x_2,N$};

				\path[->]
				(0N0N) edge node [above, sloped] {$\f$} (1F0N)
				(1F0N) edge node [above, sloped] {$u$} (1F2N)
				(1F2N) edge [loop right] node {$u$} (1F2N)
				;
			\end{tikzpicture}
			\caption{Part of $\Ver_{\Scal_7}$ corresponding to the LFSA $\Scal_7$ in Figure~\ref{fig2_response}.}
			\label{fig3_response}
		\end{center}
	\end{figure}
	The cycle $(x_1,F,x_2,N)\xrightarrow[]{u}(x_1,F,x_2,N)$ in $\Ver_{\Scal_7}$ contradicts the condition in
	Proposition~\ref{prop2_responses}, hence by Proposition~\ref{prop2_responses} $\Scal_7$ is not 
	$\fs$-diagnsoable. However, by definition, $\Scal_7$ is $\fs$-diagnosable vacuously.
\end{example}

The generalized twin plant
$$\overline{\TwPl}_{\Scal}=(Q_{\Ver},\{(\s_o,\s_o)|\s_o\in E_o\}\cup(T_{uo}\times\{\ep\})
\cup(\{\ep\}\times (E_{uo}\setminus\fs)),(x_0,N,x_0,N),\dt_{\overline{\TwPl}})$$ of $\Scal$ proposed in 
\cite{Cassez2008FaultDiagnosisStDyObser} is constructed as follows:
for all $(x_1,l_1,x_2,l_2)\in Q_{\Ver}$, $(\s_o,\s_o)$ with $\s_o\in E_o$, and $\s_{uo}\in E_{uo}\setminus
\fs$,
\begin{enumerate}[(a)]
	\item $((x_1,l_1,x_2,l_2),(\f,\ep),(x_1',F,x_2,l_2))\in\dt_{\Ver}$ if and only if $(x_1,\f,x_1')\in\dt$,
	\item $((x_1,l_1,x_2,l_2),(\s_{uo},\ep),(x_1',l_1,x_2,l_2))\in\dt_{\Ver}$ if and only if $(x_1,\s_{uo},x_1')\in\dt$,
	\item $((x_1,l_1,x_2,l_2),(\ep,\s_{uo}),(x_1,l_1,x_2',l_2))\in\dt_{\Ver}$ if and only if $(x_2,\s_{uo},x_2')\in\dt$,
	\item $((x_1,l_1,x_2,l_2),(\s_{o},\s_o),(x_1',l_1,x_2',l_2))\in\dt_{\Ver}$ if and only if $(x_1,\s_{o},x_1'),(x_2,\s_o,x_2')\in\dt$.
\end{enumerate}

There is no state of the form $(-,-,-,F)$ reachable in $\overline{\TwPl}_{\Scal}$.

\begin{proposition}[\cite{Cassez2008FaultDiagnosisStDyObser}]\label{prop3_responses}
	An $\Scal$ is not $\fs$-diagnosable if and only if in $\overline{\TwPl}_{\Scal}$ there is a reachable
	cycle in which all states are of the form $(-,F,-,N)$ and there is at least one event of the form
	$(\s,-)$ with $\s\in E$.
\end{proposition}

\begin{example}
	Consider $\Scal_7$ in Figure~\ref{fig2_response}. $\overline{\TwPl}_{\Scal_7}$ is shown as follows:
	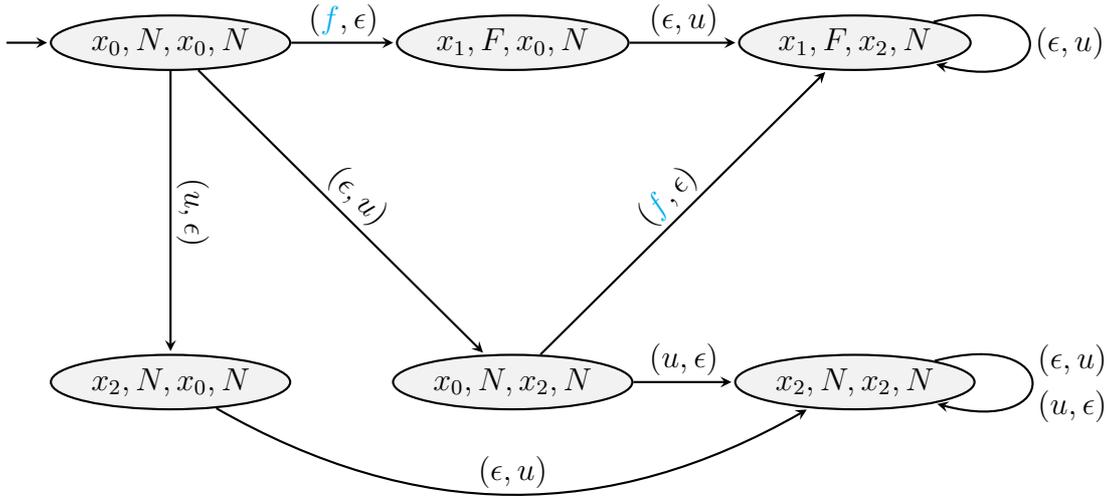
\begin{figure}[H]
		\tikzset{global scale/.style={
    scale=#1,
    every node/.append style={scale=#1}}}
		\begin{center}
			\begin{tikzpicture}[global scale = 1.0,
				>=stealth',shorten >=1pt,thick,auto,node distance=4.5 cm, scale = 1.0, transform shape,
	->,>=stealth,inner sep=2pt,
				every transition/.style={draw=red,fill=red,minimum width=1mm,minimum height=3.5mm},
				every place/.style={draw=blue,fill=blue!20,minimum size=7mm}]
				\tikzstyle{emptynode}=[inner sep=0,outer sep=0]
				\node[elliptic state, initial, initial where = left] (0N0N) {$x_0,N,x_0,N$};
				\node[elliptic state] (1F0N) [right of = 0N0N] {$x_1,F,x_0,N$};
				\node[elliptic state] (1F2N) [right of = 1F0N] {$x_1,F,x_2,N$};
				\node[elliptic state] (0N2N) [below of = 1F0N] {$x_0,N,x_2,N$};
				\node[elliptic state] (2N0N) [below of = 0N0N] {$x_2,N,x_0,N$};
				\node[elliptic state] (2N2N) [right of = 0N2N] {$x_2,N,x_2,N$};

				\path[->]
				(0N0N) edge node [above, sloped] {$(\f,\ep)$} (1F0N)
				(1F0N) edge node [above, sloped] {$(\ep,u)$} (1F2N)
				(1F2N) edge [loop right] node {$(\ep,u)$} (1F2N)
				(0N0N) edge node [above, sloped] {$(\ep,u)$} (0N2N)
				(0N2N) edge node [above, sloped] {$(\f,\ep)$} (1F2N)
				(0N0N) edge node [above, sloped] {$(u,\ep)$} (2N0N)
				(2N0N) edge [bend right] node [above, sloped] {$(\ep,u)$} (2N2N)
				(0N2N) edge node [above, sloped] {$(u,\ep)$} (2N2N)
				(2N2N) edge [loop right] node {$\begin{matrix}(\ep,u)\\(u,\ep)\end{matrix}$} (2N2N)
				;
			\end{tikzpicture}
			\caption{Part of $\overline{\TwPl}_{\Scal_7}$ corresponding to the LFSA $\Scal_7$ in Figure~\ref{fig2_response}.}
			\label{fig4_response}
		\end{center}
	\end{figure}
	In Figure~\ref{fig4_response}, there is no reachable cycle as in Proposition~\ref{prop3_responses}.
	In the cycle $(x_1,F,x_2,N)\xrightarrow[]{(\ep,u)}(x_1,F,x_2,N)$, in the unique event $(\ep,u)$,
	the left component is $\ep$. 
	Hence by Proposition~\ref{prop3_responses} $\Scal_7$ is $\fs$-diagnosable.
\end{example}

It is easy to see that the concurrent composition $\CCa(\Scal_{\fsf},\Scal_{\nsf})$ in the current paper
used for verifying diagnosability is coincidently similar to the generalized version of twin plant
$\overline{\TwPl}_{\Scal}$ proposed in 
\cite{Cassez2008FaultDiagnosisStDyObser}. After removing all $F$'s and $N$'s from $\overline{\TwPl}_{\Scal}$,
$\CCa(\Scal_{\fsf},\Scal_{\nsf})$ is obtained.


\begin{thebibliography}{10}

\bibitem{WonhamSupervisoryControl}
W.M. Wonham and K.~Cai.
\newblock {\em {Supervisory Control of Discrete-Event Systems}}.
\newblock Springer International Publishing, 2019.

\bibitem{Sipser2006TheoryofComputation}
M.~Sipser.
\newblock {\em Introduction to the Theory of Computation}.
\newblock International Thomson Publishing, 1st edition, 1996.

\bibitem{Shu2011GDetectabilityDES}
S.~Shu and F.~Lin.
\newblock Generalized detectability for discrete event systems.
\newblock {\em Systems \& Control Letters}, 60(5):310--317, 2011.

\bibitem{Jiang2001PolyAlgorithmDiagnosabilityDES}
S.~Jiang, Z.~Huang, V.~Chandra, and R.~Kumar.
\newblock A polynomial algorithm for testing diagnosability of discrete-event
  systems.
\newblock {\em IEEE Transactions on Automatic Control}, 46(8):1318--1321, Aug
  2001.

\bibitem{Yoo2002DiagnosabiliyDESPTime}
T.-S. {Yoo} and S.~{Lafortune}.
\newblock Polynomial-time verification of diagnosability of partially observed
  discrete-event systems.
\newblock {\em IEEE Transactions on Automatic Control}, 47(9):1491--1495, Sep.
  2002.

\bibitem{Genc2009PredictabilityDES}
S.~Genc and S.~Lafortune.
\newblock {Predictability of event occurrences in partially-observed
  discrete-event systems}.
\newblock {\em Automatica}, 45(2):301--311, 2009.

\bibitem{Saboori2012InfiniteStepOpacity}
A.~Saboori and C.~N. Hadjicostis.
\newblock Verification of infinite-step opacity and complexity considerations.
\newblock {\em IEEE Transactions on Automatic Control}, 57(5):1265--1269, May
  2012.

\bibitem{Yin2017TWObserverInfiniteStepOpacity}
X.~Yin and S.~Lafortune.
\newblock {A new approach for the verification of infinite-step and $K$-step
  opacity using two-way observers}.
\newblock {\em Automatica}, 80:162--171, 2017.

\bibitem{Falcone2015StrongKStepOpacity}
Y.~Falcone and H.~Marchand.
\newblock Enforcement and validation (at runtime) of various notions of
  opacity.
\newblock {\em Discrete Event Dyn.\ Sys.: Theory \& Apl.}, 25:531--570, 2015.

\bibitem{Ma2021StrongInfSODES}
Z.~Ma, X.~Yin, and Z.~Li.
\newblock Verification and enforcement of strong infinite- and $k$-step opacity
  using state recognizers.
\newblock {\em Automatica}, 133:109838, 2021.

\bibitem{Cassez2009DynamicOpcity}
F.~Cassez, J.~Dubreil, and H.~Marchand.
\newblock Dynamic observers for the synthesis of opaque systems.
\newblock In {\em Automated Technology for Verification and Analysis}, pages
  352--367, Berlin, Heidelberg, 2009. Springer Berlin Heidelberg.

\bibitem{Saboori2007CurrentStateOpacity}
A.~Saboori and C.~N. Hadjicostis.
\newblock Notions of security and opacity in discrete event systems.
\newblock In {\em 2007 46th IEEE Conference on Decision and Control}, pages
  5056--5061, Dec 2007.

\bibitem{Wu2013ComparativeAnalysisOpacity}
Y.~Wu and S.~Lafortune.
\newblock Comparative analysis of related notions of opacity in centralized and
  coordinated architectures.
\newblock {\em Discrete Event Dynamic Systems}, 23(3):307--339, Sep 2013.

\bibitem{Zhang2020DetPNFA}
K.~Zhang and A.~Giua.
\newblock {On detectability of labeled Petri nets and finite automata}.
\newblock {\em Discrete Event Dynamic Systems}, 30(3):465--497, 2020.

\bibitem{Shu2007Detectability_DES}
S.~Shu, F.~Lin, and H.~Ying.
\newblock Detectability of discrete event systems.
\newblock {\em IEEE Transactions on Automatic Control}, 52(12):2356--2359, Dec
  2007.

\bibitem{Zhang2019KDelayStrDetDES}
K.~{Zhang} and A.~{Giua}.
\newblock {$K$-delayed strong detectability of discrete-event systems}.
\newblock In {\em Proceedings of the 58th IEEE Conference on Decision and
  Control (CDC)}, pages 7647--7652, Dec 2019.

\bibitem{Moore1956}
E.F. Moore.
\newblock Gedanken-experiments on sequential machines.
\newblock {\em Automata Studies, Annals of Math. Studies}, 34:129--153, 1956.

\bibitem{Kalman1963MathDescriptionofLDS}
R.E. Kalman.
\newblock Mathematical description of linear dynamical systems.
\newblock {\em Journal of the Society for Industrial and Applied Mathematics
  Series A Control}, 1(12):152--192, 1963.

\bibitem{Model-BasedTesting2005}
M.~Broy, B.~Jonsson, J.~P. Katoen, L.~Martin, and A.~Pretschner.
\newblock {\em Model-Based Testing of Reactive Systems: Advanced Lectures
  (Lecture Notes in Computer Science)}.
\newblock Springer-Verlag New York, Inc., Secaucus, NJ, USA, 2005.

\bibitem{Wonham1985LinearMultiControl}
W.M. Wonham.
\newblock {\em Linear Multivariable Control: a Geometric Approach, 3rd Ed.}
\newblock Springer-Verlag New York, 1985.

\bibitem{Sontag1979}
E.D. Sontag.
\newblock {On the observability of polynomial systems, I: Finite-time
  problems}.
\newblock {\em SIAM Journal on Control and Optimization}, 17:139--151, 1979.

\bibitem{Conte2007AlgebraicMethodsNonlinearControlSystems}
G.~Conte, C.H. Moog, and A.M. Perdon.
\newblock {\em {Algebraic Methods for Nonlinear Control Systems, 2nd Ed.}}
\newblock Springer-Verlag London, 2007.

\bibitem{Isidori1999NonConSys}
A.~Isidori.
\newblock {\em Nonlinear Control Systems}.
\newblock Communications and Control Engineering. Springer-Verlag London, 1995.

\bibitem{Tanwani2013ObsSwiLinSys}
A.~{Tanwani}, H.~{Shim}, and D.~{Liberzon}.
\newblock {Observability for switched linear systems: characterization and
  observer design}.
\newblock {\em IEEE Transactions on Automatic Control}, 58(4):891--904, April
  2013.

\bibitem{Kibangou2016ObserNetSys}
A.~Y. Kibangou, F.~Garin, and S.~Gracy.
\newblock {Input and state observability of network systems with a single
  unknown Input}.
\newblock {\em IFAC-PapersOnLine}, 49(22):37--42, 2016.
\newblock 6th IFAC Workshop on Distributed Estimation and Control in Networked
  Systems NECSYS 2016.

\bibitem{Angulo2019StrucObserNonNetSys}
M.~T. {Angulo}, A.~{Aparicio}, and C.~H. {Moog}.
\newblock Structural accessibility and structural observability of nonlinear
  networked systems.
\newblock {\em IEEE Transactions on Network Science and Engineering}, page
  online, 2019.

\bibitem{Ramadge1986ObservabilityDES}
P.~J. Ramadge.
\newblock Observability of discrete event systems.
\newblock In {\em 1986 25th IEEE Conference on Decision and Control}, pages
  1108--1112, Dec 1986.

\bibitem{Ozveren1990ObservabilityDES}
C.~M. \"{O}zveren and A.~S. Willsky.
\newblock Observability of discrete event dynamic systems.
\newblock {\em IEEE Transactions on Automatic Control}, 35(7):797--806, Jul
  1990.

\bibitem{Zhang2017PSPACEHardnessWeakDetectabilityDES}
K.~Zhang.
\newblock {The problem of determining the weak (periodic) detectability of
  discrete event systems is PSPACE-complete}.
\newblock {\em Automatica}, 81:217--220, 2017.

\bibitem{Masopust2018ComplexityDetectabilityDES}
T.~Masopust.
\newblock Complexity of deciding detectability in discrete event systems.
\newblock {\em Automatica}, 93:257--261, 2018.

\bibitem{Zhang2020bookDDS}
K.~Zhang, L.~Zhang, and L.~Xie.
\newblock {\em Discrete-Time and Discrete-Space Dynamical Systems}.
\newblock Communications and Control Engineering. Springer International
  Publishing, 2020.

\bibitem{Sampath1995DiagnosabilityDES}
M.~Sampath, R.~Sengupta, S.~Lafortune, K.~Sinnamohideen, and D.~Teneketzis.
\newblock Diagnosability of discrete-event systems.
\newblock {\em IEEE Transactions on Automatic Control}, 40(9):1555--1575, Sep
  1995.

\bibitem{Cassez2008FaultDiagnosisStDyObser}
F.~Cassez and S.~Tripakis.
\newblock Fault diagnosis with static and dynamic observers.
\newblock {\em Fundamenta Informaticae}, 88(4):497--540, 2008.

\bibitem{Zhang2021UnifyingDetDiagPred}
K.~Zhang.
\newblock A unified method to decentralized state detection and fault
  diagnosis/prediction of discrete-event systems.
\newblock {\em Fundamenta Informaticae}, 181:339--371, 2021.

\bibitem{Mazare2004Opacity}
L.~Mazar\'{e}.
\newblock Using unification for opacity properties.
\newblock In {\em Proceedings of the Workshop on Issues in the Theory of
  Security (WITS'04)}, pages 165--176, 2004.

\bibitem{Chaum1988DiningCryptographers}
D.~Chaum.
\newblock The dining cryptographers problem: Unconditional sender and recipient
  untraceability.
\newblock {\em Journal of Cryptology}, 1(1):65--75, 1988.

\bibitem{Saboori2010PhDThesisOpacity}
A.~Saboori.
\newblock {\em Verification and Enforcement of State-Based Notions of Opacity
  in Discrete Event Systems}.
\newblock PhD thesis, University of Illinois at Urbana-Champaign, 2010.

\bibitem{Wu2014PhDThesisOpacity}
Y.~Wu.
\newblock {\em Verification and Enforcement of Opacity Security Properties in
  Discrete Event Systems}.
\newblock PhD thesis, University of Michigan, 2014.

\bibitem{Wu2018SyntheObfuscation}
Y.~Wu, V.~Raman, B.C. Rawlings, S.~Lafortune, and S.A. Seshia.
\newblock Synthesis of obfuscation policies to ensure privacy and utility.
\newblock {\em Journal of Automated Reasoning}, 60(1):107--131, 2018.

\bibitem{Wu2016ObfuscatorSynthesis}
Y.~Wu, V.~Raman, S.~Lafortune, and S.A. Seshia.
\newblock Obfuscator synthesis for privacy and utility.
\newblock In Sanjai Rayadurgam and Oksana Tkachuk, editors, {\em NASA Formal
  Methods}, pages 133--149, Cham, 2016. Springer International Publishing.

\bibitem{Goes20181IndoorPrivacyObsfucation}
R.M. G\'{o}es, B.C. Rawlings, N.~Recker, G.~Willett, and S.~Lafortune.
\newblock Demonstration of indoor location privacy enforcement using
  obfuscation.
\newblock {\em IFAC-PapersOnLine}, 51(7):145--151, 2018.
\newblock 14th IFAC Workshop on Discrete Event Systems WODES 2018.

\bibitem{Bryans2008OpacityTransitionSystems}
J.~W. Bryans, M.~Koutny, L.~Mazar\'{e}, and P.~Y.~A. Ryan.
\newblock Opacity generalised to transition systems.
\newblock {\em International Journal of Information Security}, 7(6):421--435,
  Nov 2008.

\bibitem{Lin2011OpacityDES}
F.~Lin.
\newblock Opacity of discrete event systems and its applications.
\newblock {\em Automatica}, 47(3):496--503, March 2011.

\bibitem{Saboori2013InitialStateOpacity}
A.~Saboori and C.~N. Hadjicostis.
\newblock Verification of initial-state opacity in security applications of
  discrete event systems.
\newblock {\em Information Sciences}, 246:115--132, 2013.

\bibitem{Saboori2009KStepOpacityFA}
A.~{Saboori} and C.~N. {Hadjicostis}.
\newblock Verification of {$K$}-step opacity and analysis of its complexity.
\newblock In {\em Proceedings of the 48h IEEE Conference on Decision and
  Control (CDC) held jointly with 2009 28th Chinese Control Conference}, pages
  205--210, 2009.

\bibitem{Balun2021ComparingOpacityinDES}
J.~Balun and T.~Masopust.
\newblock Comparing the notions of opacity for discrete-event systems.
\newblock {\em Discrete Event Dynamic Systems}, 2021.

\bibitem{Cassandras2009DESbook}
C.G. Cassandras and S.~Lafortune.
\newblock {\em Introduction to Discrete Event Systems}.
\newblock Springer Publishing Company, Incorporated, 2nd edition, 2010.

\bibitem{Balun2021KSO_ConCurren}
J.~Balun and T.~Masopust.
\newblock K-step opacity in discrete event systems: Verification, complexity,
  and relations.
\newblock \url{https://arxiv.org/abs/2109.02158}.

\bibitem{Han2021StrongOpacityDES}
X.~Han, K.~Zhang, J.~Zhang, Z.~Li, and Z.~Chen.
\newblock Strong current-state and initial-state opacity of discrete-event
  systems.
\newblock \url{https://arxiv.org/abs/2109.05475}.

\end{thebibliography}
\end{document}